\documentclass[lettersize, journal]{IEEEtran}
\usepackage{setspace}
\usepackage{cite}
\usepackage{pgfplots}

\pgfplotsset{compat=1.10}

\usepgfplotslibrary{fillbetween}

\usepackage{bm}
\usepackage{lipsum}
\usepackage{makeidx}
\usepackage{enumerate}
\usepackage{color}
\usepackage{cite}
\usepackage{amsmath,amsthm}   
\usepackage{amssymb}
\usepackage{nomencl}
\usepackage{multirow}
\usepackage{graphicx}
\usepackage{epstopdf}
\usepackage{threeparttable}
\usepackage{multicol}
\DeclareGraphicsExtensions{.pdf,.jpeg,.png,.jpg,.emf,.eps}
\usepackage{calc}
\usepackage{here}
\usepackage{url}

\usepackage{upref}
\usepackage{comment}
\usepackage{times}
\usepackage{dsfont}
\usepackage{epic,eepic}
\usepackage{rawfonts}
\usepackage[T1]{fontenc}
\usepackage{latexsym}
\usepackage{amsfonts}

\usepackage{capitalgreekitalic}

\usepackage[font=small,labelfont=bf]{caption}
\usepackage[font=small,labelfont=bf]{subcaption}

\usepackage{xcolor}
\usepackage{tikz,pgfplots}
\usetikzlibrary{calc}
\usetikzlibrary{intersections}
\usetikzlibrary{arrows,shapes}
\usetikzlibrary{positioning}

\newtheorem{corollary}{Corollary}

\newtheorem{lemma}{Lemma}

\newtheorem{remark}{Remark}

\theoremstyle{definition}

\allowdisplaybreaks

\begin{document}
\title{Optimization of the Downlink Spectral- and Energy-Efficiency of  RIS-aided Multi-user \\URLLC MIMO Systems
}
\author{Mohammad Soleymani, \emph{Member, IEEE},  
Ignacio Santamaria, \emph{Senior Member, IEEE}, 
\\Eduard Jorswieck, \emph{Fellow, IEEE},
Robert Schober, \emph{Fellow, IEEE}, \\and
Lajos Hanzo, \emph{Life Fellow, IEEE}
 \\ \thanks{ 
Mohammad Soleymani is with the Signal and System Theory Group, Universit\"at Paderborn, Germany (email: \protect\url{mohammad.soleymani@uni-paderborn.de}).  

Ignacio Santamaria is with the Department of Communications Engineering, Universidad de Cantabria, Spain (email: \protect\url{i.santamaria@unican.es}).

Eduard Jorswieck is with the Institute for Communications Technology, Technische Universit\"at Braunschweig, 38106 Braunschweig, Germany
(email: \protect\url{jorswieck@ifn.ing.tu-bs.de}).

Robert Schober is with the Institute for Digital Communications, Friedrich Alexander University of Erlangen-Nuremberg, Erlangen 91058, Germany (email: \protect\url{robert.schober@fau.de}).

Lajos Hanzo is with the Department of Electronics and Computer Science, University of Southampton, Southampton, United Kingdom (email:
\protect\url{lh@ecs.soton.ac.uk}).

The work of Ignacio Santamaria was funded by MCIN/ AEI /10.13039/501100011033, under Grant PID2022-137099NB-C43 (MADDIE), and by European Union's (EU's) Horizon Europe project 6G-SENSES under Grant 101139282. The work of Eduard Jorswieck was supported by the Federal Ministry of Education and Research (BMBF, Germany) through the Program of Souver\"an. Digital. Vernetzt. joint Project 6G-RIC, under Grant 16KISK031, and by European Union's (EU's) Horizon Europe project 6G-SENSES under Grant 101139282. Lajos Hanzo would like to acknowledge the financial support of the Engineering and Physical Sciences Research Council (EPSRC) projects under grant EP/Y037243/1, EP/W016605/1, EP/X01228X/1, EP/Y026721/1,
EP/W032635/1 and EP/X04047X/1 as well as of the European Research Council’s Advanced Fellow Grant QuantCom (Grant No. 789028). The work of Robert Schober was partly supported by the Federal Ministry of Education and Research of Germany under the programme of ``Souveran. Digital. Vernetzt.'' joint project 6G-RIC (project identification number: PIN 16KISK023) and by the Deutsche Forschungsgemeinschaft (DFG, German Research Foundation) under grant SCHO 831/15-1.
}}
\maketitle
\begin{abstract}
Modern wireless communication systems are expected to provide improved latency and reliability. To meet these expectations, a short packet length is needed, which makes the first-order Shannon rate an inaccurate performance metric for such communication systems. A more accurate approximation of the achievable rates of  finite-block-length (FBL) coding regimes is known as the normal approximation (NA). It is therefore of substantial interest to study the optimization of the FBL rate in multi-user multiple-input multiple-output (MIMO) systems, in which each user may transmit and/or receive multiple data streams.  Hence, we formulate a general optimization problem for improving the spectral and energy efficiency of multi-user MIMO-aided ultra-reliable low-latency communication (URLLC) systems, which are assisted by reconfigurable intelligent surfaces (RISs). We show that an RIS is capable of substantially improving the performance of multi-user MIMO-aided URLLC systems. Moreover, the benefits of RIS increase as the packet length and/or the tolerable bit error rate are reduced. This reveals that RISs can be even more beneficial in URLLC systems for improving the FBL rates than in conventional systems approaching Shannon rates.
\end{abstract} 
\begin{IEEEkeywords}
Energy efficiency,
 MIMO broadcast channels,  reconfigurable intelligent surface, spectral efficiency, ultra-reliable low-latency communication.
\end{IEEEkeywords}
\section{Introducton}\label{1}
The sixth generation (6G) of wireless communication systems is expected to significantly improve the spectral efficiency (SE), energy efficiency (EE), and reliability of the existing systems, despite of providing a lower latency than 5G \cite{wang2023road, gong2022holographic}. Thus, 6G should employ radical new technologies such as reconfigurable intelligent surface (RIS) to meet these expectations. Moreover, 6G has to support a large variety of applications, which require ultra-reliable and low-latency communications (URLLC) \cite{wang2023road, gong2022holographic}.
To attain low latency, realistic finite block length (FBL) codes have to be employed. In this content, the classic Shannon rate is an inaccurate performance metric for URLLC systems. Indeed, the FBL rate is more challenging to optimize than the Shannon rate, especially in  multiple-input multiple-output (MIMO) systems, when multi-stream communication is targeted. 
In fact, to the best of our knowledge,  resource allocation has not been designed for multi-user MIMO (MU-MIMO) systems relying on FBL coding in the open literature for the scenario of multiple streams per user.
Developing suitable resource allocation schemes is even more challenging in RIS-assisted systems since this requires the joint optimization of  the transmit covariance matrices and the channels, which depend on the RIS elements.
To close this knowledge gap,  we derive a closed-form expression for the rates of users in MU-MIMO systems using realistic FBL coding when multiple streams are allowed. Then, we develop an optimization framework for MU-MIMO RIS-assisted URLLC systems and show that an RIS can substantially improve the SE and EE. Our results show that an RIS can be even more beneficial in MIMO URLLC systems than in systems approaching the classic Shannon rate, since an RIS provides higher gains for short  packets and/or for low tolerable bit error rates (BERs).

\subsection{Literature Review}\label{sec-i-a}
A main goal of 6G is to drastically enhance the SE and EE, which are even more vital for applications related to FBL coding. To realize this goal, 6G has to employ powerful emerging technologies such as RISs as well as existing  MIMO solutions \cite{wu2021intelligent, di2020smart}. An RIS was shown to be able to substantially enhance the EE and SE \cite{huang2019reconfigurable, wu2019intelligent, kammoun2020asymptotic,  soleymani2022noma,  santamaria2023interference, soleymani2023maximization}, when studying different performance metrics such as the sum rate, minimum rate, total power consumption required for achieving a specific quality of service (QoS), minimum signal-to-interference-plus-noise ratio (SINR), global EE (GEE), and interference leakage. 
For instance in \cite{huang2019reconfigurable}, the authors proposed algorithms for increasing the GEE and sum rate of a multiple-input single-output (MISO) broadcast channel (BC). In \cite{wu2019intelligent}, it was shown that an RIS reduces the power consumption of the single-cell MISO BC when a minimum SINR per user has to be ensured. Moreover, the authors of \cite{kammoun2020asymptotic} showed that an RIS is capable of increasing the minimum SINR of a single-cell MISO BC for a given power budget. The benefit of RIS in terms of both the minimum rate and the minimum EE was studied in \cite{soleymani2022noma} for a multi-cell MISO BC employing non-orthogonal multiple access (NOMA). In a multi-cell BC, each user is associated with only a single base station (BS) at a time. In \cite{xu2023algorithm, yao2023robust, lyu2023energy}, it was shown that RISs can also improve the performance of cell-free systems in which each user simultaneously communicates with several access points, rather than a single BS.

The  performance of RIS relying on the Shannon capacity achieving codes  has been studied in \cite{huang2019reconfigurable, wu2019intelligent, kammoun2020asymptotic,  soleymani2022noma,  santamaria2023interference}, but naturally, an RIS can also be beneficial in multi-user systems using FBL coding, as shown in \cite{soleymani2023spectral, soleymani2023optimization, li2021aerial , vu2022intelligent, xie2021user,     almekhlafi2021joint}. For instance, in \cite{soleymani2023spectral},  resource allocation schemes were developed for MISO URLLC systems, assisted by simultaneously transmit and reflect (STAR-) RIS, and it was shown that an RIS (either STAR or purely reflective) substantially improves the SE and EE of URLLC systems. In \cite{soleymani2023optimization}, it was demonstrated that RIS and rate splitting can be mutually beneficial tools of enhancing the EE and SE performance of interference-limited URLLC systems. In \cite{vu2022intelligent}, the advantage of employing an RIS and NOMA in a two-user single-input single-output (SISO) URLLC BC was shown.  

The papers \cite{soleymani2023spectral, soleymani2023optimization, li2021aerial , vu2022intelligent, xie2021user, almekhlafi2021joint} studied multi-user RIS-assisted URLLC systems, but only supported single-stream data transmission per user. However, RIS can also improve the system performance when parallel frequency-domain channels are employed along with FBL coding \cite{ghanem2022optimal}. Nevertheless, resource allocation has not been designed for multi-user MIMO URLLC systems supporting multiple streams per user in the open literature. This is particularly challenging for RIS-assisted systems. 
Indeed, only a few treatises exist in multiple-stream data transmission in MIMO systems, which mainly studied a single-user scenario without considering RISs \cite{ makki2018finite, makki2016required,  li2022ultra}. Thus, multi-user MIMO systems both with and without RISs require further investigations. MIMO systems support multiple-stream data transmissions per user, which exploit the spectrum  efficiently and improve the EE at a specific QoS. Below, we briefly describe the challenges of optimizing the FBL rates when multiple-stream data transmission is supported. 

In the FBL regime, the achievable rate depends not only on the Shannon rate, $C$, but also on the channel's dispersion, $V$, the packet length, $n_t$, and the tolerable bit error rate $\epsilon$. An accurate approximation for the achievable rate of FBL coding in parallel channels is the normal approximation (NA)\footnote{Note that the NA may not be accurate when the packet length is extremely short, and/or the tolerable bit error rate is extremely low. For further discussions regarding the accuracy of the NA, please refer to \cite{polyanskiy2010channel, erseghe2016coding, erseghe2015evaluation}.}, which is given by \cite[Theorem 78]{polyanskiy2010}
\begin{equation}
r=\sum_{i=1}^IC_i-Q^{-1}(\epsilon)\sqrt{\frac{\sum_{i=1}^IV_i}{n_t}},
\end{equation}
where $Q^{-1}$ is the inverse of the well-known $Q$ function of Gaussian distributions, $I$ is the number of parallel channels, $C_i$ and $V_i$ are, respectively, the Shannon rate and the channel dispersion of the $i$-th parallel channel. Note that the channel dispersion and Shannon rate of parallel channels are, respectively, a summation of the channel dispersions and Shannon rates of all individual channels, i.e., $C=\sum_{i=1}^IC_i$ and $\sum_{i=1}^IV_i$.   {The Shannon rate of MIMO systems can  also be represented in a closed-form matrix format, and there are already existing contributions on optimizing the Shannon rates in MIMO RIS-assisted systems \cite{soleymani2022improper, soleymani2023rate}. However, the achievable channel  dispersion term for Gaussian signals has a fractional structure, which is more challenging to optimize, and its closed-form matrix format has not been derived in the related works.  
Hence, resource allocation for parallel channels relying on FBL coding can be much more complicated than for single-stream channels. Moreover, the channel's dispersion term makes it impossible to reuse the existing solutions for MIMO RIS-assisted systems, when FBL coding is employed.} In the next subsection, we provide a critical review of the existing works on RIS-assisted URLLC systems and discuss the open topics that merit further investigations. 

\begin{table}
\centering
\scriptsize
\caption{Overview of most closely related works on RIS-assisted URLLC systems.}\label{table-1}
\begin{tabular}{|c|c|c|c|c|c|c|c|c|c|c|c|c|c|c|}
 \hline	
 &This paper&\cite{soleymani2023optimization}&\cite{abughalwa2022finite}&\cite{li2021aerial,  vu2022intelligent, xie2021user,     almekhlafi2021joint, ghanem2022optimal, ghanem2021joint}&\cite{ ren2021intelligent,   zhang2021irs}\\
 \hline
Multi-user&  $\surd$&$\surd$&  $\surd$&$\surd$&-
\\
 \hline
 Ch. disp. in \cite{scarlett2016dispersion}&  $\surd$&$\surd$&$\surd$&-&-   \\
\hline
EE&  $\surd$&$\surd$&-&-&-
\\
\hline
MIMO&
$\surd$&-&-&-&-
\\
\hline
Multiple streams&
$\surd$&-&-&-&-
\\
\hline
STAR-RIS&$\surd$&-&-&-&-
\\
\hline
 {Transmission delay}&$\surd$&-&-&-&-
\\
\hline
		\end{tabular}
\normalsize
\end{table}

\subsection{Motivation}
The most closely related treatises on RIS-assisted URLLC system designs are compared in Table \ref{table-1}, based on the system model, network scenario, performance metrics, and the channel dispersion encountered in multi-user systems. As shown in the table,  most of the studies on FBL transmission in RIS-assisted systems have focused on SISO/MISO systems, when only a single-stream data transmission per user is allowed.  Additionally, there is a limited number of contributions on EE in RIS-assisted URLLC systems and EE metrics have not been studied in multi-user MIMO systems using FBL coding. Note that in URLLC systems, the EE can be even more vital, since in some applications, it might be impossible to replace the battery of users/nodes, and consequently, the network must be as energy efficient as possible.

Moreover,  there is no work on multi-user MIMO systems with FBL considering the achievable channel dispersion term in \cite{scarlett2016dispersion}, even for systems without RIS. It should be emphasized that Gaussian signaling cannot achieve the optimal channel dispersion in the presence of interference in multi-user systems. Hence, when employing Gaussian signaling, it is more accurate to consider  the suboptimal channel dispersion term in \cite{scarlett2016dispersion} instead of the optimal one since  the channel dispersion in \cite{scarlett2016dispersion} can be attained through Gaussian signals  in the presence of both interference and white additive Gaussian noise in multi-user systems.

To sum up, resource allocation schemes should be developed for multi-user  RIS-assisted URLLC  MIMO systems with more emphasis on EE and by allowing multiple-data-stream transmissions per user and considering the achievable channel dispersion term for Gaussian signals. Thus, a general optimization framework for multi-user RIS-assisted MIMO systems with FBL coding can facilitate future studies in this field.

\subsection{Contribution}
We propose a general optimization framework  for maximizing the SE and EE of multi-user MIMO RIS-aided URLLC systems. To the best of our knowledge, this is the first paper on the resource allocation of MU-MIMO RIS-aided URLLC systems supporting multiple stream per users.   {To develop a general framework, we first formulate a closed-form expression for the channel dispersion of MIMO systems in the presence of interference, based on both the optimal  channel dispersion as well as on the  channel dispersion term in  \cite{scarlett2016dispersion}.} Then, we propose specific schemes for optimizing the FBL rate expressions by employing majorization minimization (MM), alternating optimization (AO), and fractional programming (FP) tools such as Dinkelbach-based algorithms. As indicated in Section \ref{sec-i-a}, due to the channel dispersion term, which has a fractional structure, the FBL rates are much more challenging to optimize than the classic Shannon rates. Moreover, it is impossible to adapt the established works on MIMO RIS-assisted systems with Shannon rates to the systems with FBL coding, and to evaluate how the reliability and latency constraints influence the effectiveness of RISs.  Thus, the main novelty of this treatise is the derivation of closed-form expressions for the channel dispersion, followed by the  development of algorithms to optimize over the FBL rates, including the dispersion term.

To elaborate, our optimization framework is flexible and can be utilized in a wide range of MU-MIMO URLLC systems aided by RISs. Additionally, the framework may be used for solving a broad spectrum of optimization problems for which the objective function and/or constraints can be, but are not limited to, linear functions of the rates and/or EE of users. The convergence of our framework is ensured towards a stationary point for the general optimization problem, when the feasibility set for the RIS elements conforms to a convex set. We consider a multi-cell MIMO BC as an example of the networks that our framework can be applied to.  {Furthermore, we consider both the EE and SE metrics as well as the transmission delay for investigating the performance of RIS in MU-MIMO URLLC systems.} For the SE metric, we consider the sum rate and the minimum rate of the users, which are among the  most common performance metrics for SE.  Moreover, we evaluate the EE by optimizing the GEE and the minimum EE of users. Note that the sum rate and global EE are pivotal overall system performance metrics. By contrast, the minimum rate and the EE of specific users consider the individual performance of the users and can provide reasonable rate/EE-fairness among the users since typically all the users are allocated similar rate/EE when the minimum rate/EE  is maximized. Thus, considering all these metrics can provide a complete picture of the performance of MU-MIMO URLLC systems aided by RISs. Moreover, we make realistic assumptions regarding the channel models and the feasible sets of the RIS coefficients for appropriately examining the RIS performance.

In addition to passive and reflective RISs, we also consider simultaneously transmitting and reflecting (STAR) RIS, which provides a full $360^\circ$ coverage. Moreover, we show that RISs can significantly enhance the EE and SE of a multi-user MIMO URLLC BC. Notably, the advantages of RISs escalate with shorter packet lengths and/or more stringent reliability constraints. This implies that the benefits of RIS can be higher in MU-MIMO URLLC systems. However, it should be noted that the performance of the system may become degraded, if the RIS elements are inaccurately optimized. 

\subsection{Organization and Notations}
The structure of the paper is outlined in the following. Section \ref{sec=ii} describes our network scenario, RIS model, and signal model as well as the rate and EE expressions. Moreover, in Section \ref{sec=ii}, we formulate the  optimization problem considered. Section \ref{sec-iii} presents the optimization framework proposed. Section \ref{sec-iv} presents our numerical results. Finally, Section \ref{sec-v} concludes the paper. 

The trace and determinant of the matrix ${\bf X}$ are, respectively, denoted as $\text{Tr}({\bf X})$ and $|{\bf X}|$. We represent the conjugate of complex variable/vector/matrix $x/{\bf x}/{\bf X}$ as $x^*/{\bf x}^*/{\bf X}^*$. The mathematical expectation is denoted as $\mathbb{E}\{\cdot\}$. The identity matrix is represented as ${\bf I}$. Moreover, $\mathcal{O}[\cdot]$ is the big-O notation for representing the computational complexity of algorithms.
\section{System model}\label{sec=ii}
 Our proposed framework can be applied to a large family of RIS-assisted MU-MIMO URLLC systems that treat interference as noise at the receivers. As an example of such MU-MIMO systems, we consider a multi-cell MIMO RIS-assisted downlink (DL) BC comprising $L$ BSs, as shown in Fig. \ref{Fig-sys-model}. We assume that BS $l$ has $N_{BS,l}$ DL transmission antennas (TAs) and serves $K_l$ multiple-antenna users. The $k$-th user associated to BS $l$, denoted as U$_{lk}$, has $N_{u,lk}$ receive antennas (RAs). Additionally, we assume that there are $M$ reflective passive RISs to assist the BSs, and the $m$-th RIS has $N_{RIS,m}$ elements.
 {Furthermore, we assume perfect, instantaneous, and global CSI, consistent with many other studies on RISs \cite{ huang2019reconfigurable, wu2019intelligent, soleymani2022noma, pan2020multicell, soleymani2022improper, soleymani2022rate}. This assumption is also commonly used in the development of resource allocation solutions for URLLC systems \cite{nasir2020resource, wang2023flexible, ghanem2021joint, pala2022joint, ghanem2020resource, soleymani2023spectral}. These solutions are particularly applicable in systems with large channel coherence time, where the channel state remains stable for extended periods, making channel estimation easier and more accurate. In such systems, resource allocation solutions can be reused across multiple time slots, and the pilot overhead required for acquiring CSI and computing the solutions is relatively low. Additionally, investigating the performance of RISs under the assumption of perfect CSI helps to illustrate the essential tradeoffs in system design and provides an upper bound of the system performance.}

\begin{figure}[t!]
    \centering
\includegraphics[width=.45\textwidth]{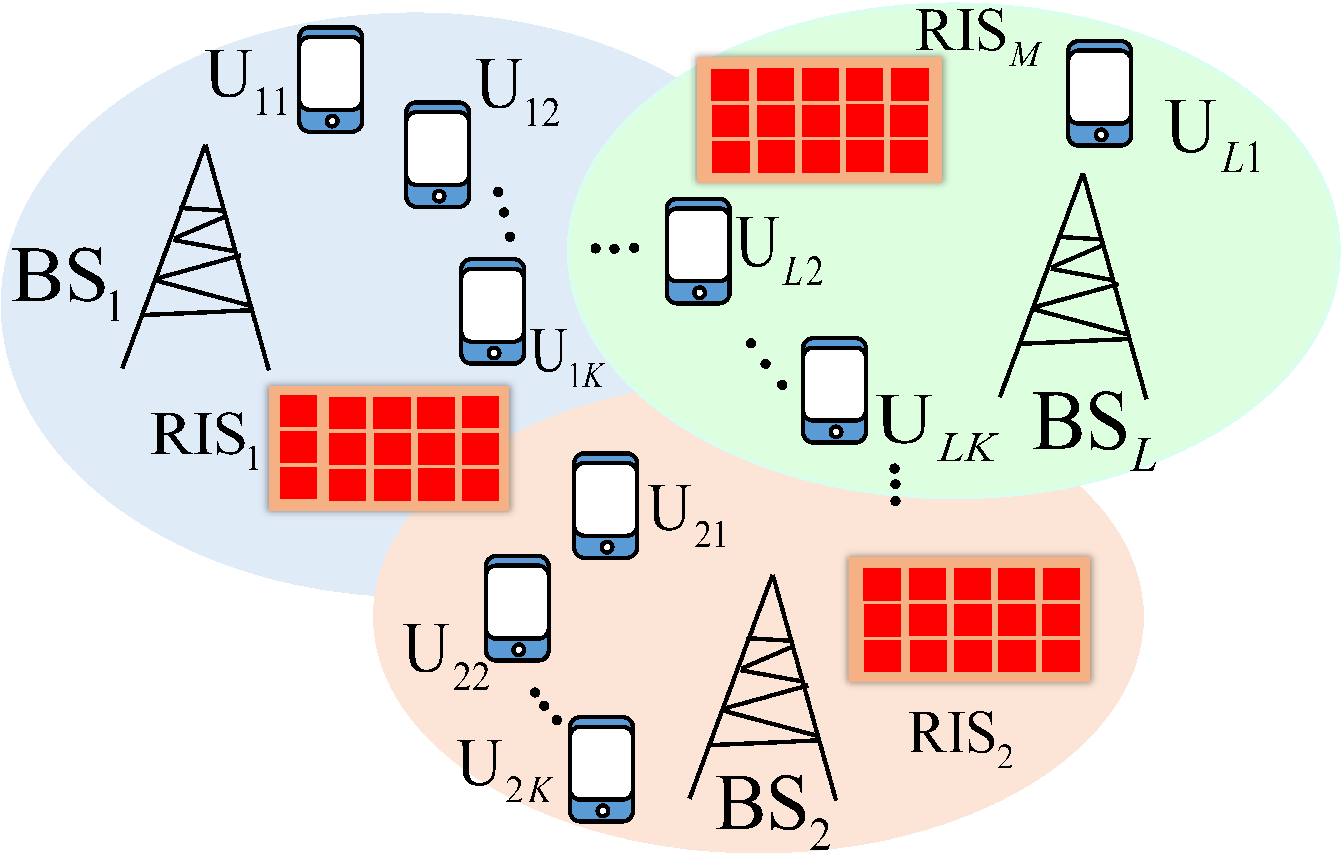} 
     \caption{A multi-cell BC assisted by RISs.}
	\label{Fig-sys-model}
\end{figure}
\subsection{RIS Model}

 {We consider two nearly-passive RIS architectures, namely reflective RIS and STAR-RIS, and employ the RIS model of \cite{pan2020multicell} for the MIMO multi-cell BC.}
\subsubsection{ {Reflective RIS}}
The channel matrix between BS $i$ and U$_{lk}$ as a function of the RIS matrices is given by
 \begin{equation}\label{ch-equ}
\mathbf{H}_{lk,i}\!\left(\{\bf{\Psi}\}\right)\!\!=\!\!\! 
\underbrace{\sum_{m=1}^M\!\!\mathbf{G}_{lk,m}{\bf \Psi}_m\mathbf{G}_{m,i}}_{\text{Link through RIS}}
+\!\!\!\!
\underbrace{\mathbf{F}_{lk,i}}_{\text{Direct link}}
\!\!\!\!
\in\mathbb{C}^{N_{u,lk}\times N_{BS,l}}
\!,
\end{equation}
where 
$\mathbf{F}_{lk,i}\in\mathbb{C}^{N_{u,lk}\times N_{BS,i}}$ is the channel matrix between the $i$-th BS and U$_{lk}$,  
$\mathbf{G}_{lk,m}\in\mathbb{C}^{N_{u,lk}\times N_{RIS,m}}$ is the channel matrix between the $m$-th RIS and U$_{lk}$, and 
$\mathbf{G}_{m,i}\in\mathbb{C}^{N_{RIS,m}\times N_{BS,i}}$ is the channel matrix between the $i$-th BS and the $m$-th RIS. Additionally, $\{{\bf\Psi}\}=\{{\bf\Psi}_m\}_{m=1}^M$ denotes the set of all coefficients of RISs, where ${\bf\Psi}_m\in\mathbb{C}^{N_{RIS,m}\times N_{RIS,m}}$ is a diagonal matrix, containing the vector of reflecting coefficients of the $m$-th RIS 
\begin{equation*}
{\bf\Psi}_m=\text{diag}\left(\psi_{m_1}, \psi_{m_2},\cdots,\psi_{m_{N_{RIS,m}}}\right).
\end{equation*}
Assuming having nearly passive RISs, the absolute value of the RIS coefficients cannot be greater than 1, which results in the following set for the feasible RIS coefficients
\cite[Eq. (11)]{wu2021intelligent}
\begin{equation}
\mathcal{T}_{U}=\left\{\psi_{m_n}:|\psi_{m_n}|^2\leq 1 \,\,\,\forall m,n\right\}.
\end{equation}
In this feasibility set, the amplitude and phase of each RIS element are assumed to be independent optimization variables, which might not be realistic. 
Another common assumption is that the RIS coefficients have to adhere to the unit modulus constraint \cite{di2020smart, wu2021intelligent, wu2019intelligent, kammoun2020asymptotic, yu2020joint, pan2020multicell, zhang2020intelligent}, which leads to 
\begin{equation}
\mathcal{T}_{I}=\left\{\psi_{m_n}:|\psi_{m_n}|= 1 \,\,\,\forall m,n\right\}.
\end{equation}
In this feasibility set, the amplitude of each RIS coefficient is assumed to be equal to 1, while the phases can be optimized. As $\mathcal{T}_{I}\subset\mathcal{T}_{U}$, it can be expected that the algorithms for $\mathcal{T}_{U}$ outperform the algorithms for $\mathcal{T}_{I}$.  

\subsubsection{ {STAR-RIS}}
STAR-RIS provides an omni-directional $360^\circ$ full-place coverage. In STAR-RIS, each component can operate in both reflection and transmission mode \cite{zhang2022intelligent, liu2021star}. Thus, there are two complex-valued optimization parameters per element, when STAR-RIS is employed. We denote the reflection/transmission coefficient for the $n$-th element of the $m$-th RIS as $\psi_{m_n}^r/\psi_{m_n}^t$. Based on the position of the user with respect to STAR-RIS, the STAR-RIS can optimize the channel of the user only through the reflection or transmission coefficients. Therefore,  
the channel between BS $i$ and U$_{lk}$ is
\begin{equation}
\mathbf{H}_{lk,i}\left(\{\bf{\Psi}\}\right)=
{\sum_{m=1}^M\mathbf{G}_{lk,m}{\bf \Psi}_m^{t/r}\mathbf{G}_{m,i}}
+
{\mathbf{F}_{lk,i}},
\end{equation}
where we have ${\bf\Psi}_m^r=\text{diag}\left(\psi_{m_1}^r, \psi_{m_2}^r,\cdots,\psi_{m_{N_{RIS,m}}^r}\right)$
and ${\bf\Psi}_m^t=\text{diag}\left(\psi_{m_1}^t, \psi_{m_2}^t,\cdots,\psi_{m_{N_{RIS,m}}^t}\right)$. 
Assuming operating in a passive mode, the absolute values of the reflection and transmission coefficients have to satisfy
\begin{equation}
|\psi_{m_n}^r|^2+|\psi_{m_n}^t|^2\leq 1, \hspace{1cm}\forall m,n,
\end{equation}
which yields the set
\begin{equation}
\mathcal{T}_{SU}=\left\{\psi_{m_n}^r,\psi_{m_n}^t:|\psi_{m_n}^r|^2+|\psi_{m_n}^t|^2\leq 1 \,\,\,\forall m,n\right\}.
\end{equation}
Assuming operating in the passive mode with equal input and output powers, we have
\begin{equation}\label{eq-8}
|\psi_{m_n}^r|^2+|\psi_{m_n}^t|^2= 1, \hspace{1cm}\forall m,n,
\end{equation}
which results in
\begin{equation}
\mathcal{T}_{SI}=\left\{\psi_{m_n}^r,\psi_{m_n}^t:|\psi_{m_n}^r|^2+|\psi_{m_n}^t|^2= 1 \,\,\,\forall m,n\right\}.
\end{equation}
There are three different STAR-RIS schemes, including the energy splitting (ES), mode switching (MS) and time switching (TS) schemes \cite{mu2021simultaneously, liu2021star}. 
Since the main focus of this work is on evaluating the impact of employing multiple streams per user regimes, we consider only the MS scheme. Note that the MS scheme has a lower implementation complexity than the ES scheme, but its performance may be comparable to that of the ES scheme as shown in, e.g., \cite{soleymani2024maximizing, soleymani2023energy, soleymani2023noma}. 
 The  framework proposed in this treatise can be extended to include the ES and TS schemes by following an approach similar to \cite{soleymani2023energy, soleymani2023noma}.

\subsubsection{ {Brief comparison of reflective RIS and STAR-RIS}}

 {The main difference between a reflective RIS and a STAR-RIS is in the coverage area of these RIS architectures, which makes each suitable for a different set of applications. STAR-RIS can provide omni-directional coverage, while reflective RIS can assist the communication between a BS and a user only if they are in the reflection space of the RIS. Thus, when the RIS can be positioned for ensuring that all the transceivers are in the reflection space of the RIS, a reflective RIS could be a more suitable option. However, when the BS is located outdoors and supports both indoor and outdoor users, STAR-RISs are preferable.}

 {From an optimization point of view, each STAR-RIS element has two complex-valued coefficients. Thus, if the ES scheme is employed,  algorithms conceived for STAR-RIS could have slightly higher computational complexities compared to reflective RIS.  However, if MS or TS schemes are utilized, each RIS element operates either in the reflection mode or in the transmission mode. Thus, only one coefficient for each STAR-RIS element is  optimized, which reduces the computational complexity to the same order as that of reflective RIS algorithms.
}

 { 
The channel matrices are assumed to be linear/affine functions of the RIS elements in both reflective RIS and STAR-RIS. To simplify the notations/equations, we remove this dependency and subsequently denote the channels as $\mathbf{H}_{lk,i}$ for all $l,k,i$, hereafter. Additionally, we denote the set of the feasible RIS elements as $\mathcal{T}$, unless we  refer to a specific set. Note that there are other RIS technologies and/or more practical feasibility sets for STAR/reflective RIS as mentioned in \cite{wu2021intelligent, 9774942}, which should be considered in future studies.}

\subsection{Signal Model}
 
We assume that BS $l$ broadcasts the signal
 \begin{equation}
{\bf x}_l=\sum_{k=1}^{K_l}{\bf x}_{lk} \in \mathbb{C}^{ N_{BS,l}\times 1},
\end{equation}
where ${\bf x}_{lk}$ is the signal intended for user U$_{lk}$, which is a zero-mean complex Gaussian random vector with covariance ${\bf P}_{lk}=\mathbb{E}\{{\bf x}_{lk}{\bf x}_{lk}^H\}$, where $\mathbb{E}\{{\bf x}\}$ is the mathematical expectation of ${\bf x}$. We assume that the zero-mean signals  ${\bf x}_{lk}$ are independent from each other, i.e., $\mathbb{E}\{{\bf x}_{lk}{\bf x}_{ij}^H\}={\bf 0}$ for $i\neq l$ and/or $j\neq k$. Additionally, we denote the covariance matrix of ${\bf x}_l$ by ${\bf P}_{l}=\mathbb{E}\{{\bf x}_{l}{\bf x}_{l}^H\}$. Since the signals ${\bf x}_{lk}$ are zero-mean and independent random vectors for all $l$ and $k$, we have $\mathbf{P}_{l} =\sum_{k}\mathbf{P}_{lk}$. The set containing all the feasible transmit covariance matrices is denoted as $\mathcal{P}$ and it is given by
\begin{equation}
\mathcal{P}=
\left\{
\mathbf{P}_{lk}:
\text{Tr}\left(\mathbf{P}_{l}\right)\leq p_l, \mathbf{P}_{lk}\succcurlyeq\mathbf{0},\,\, \forall l,k\right\},
\end{equation}
where $p_l$ is the power budget of the $l$-th BS.

The received signal at U$_{lk}$ is given by
\begin{align}\nonumber
{\bf y}_{lk}&=\sum_i{\bf H}_{lk,i}{\bf x}_i+{\bf n}_{lk},\\
&=\!\!
\underbrace{
{\bf H}_{lk,l}{\bf x}_{lk}
}_{\text{Desired Signal}}
\!+\!\!
\underbrace{
\sum_{j\neq k}{\bf H}_{lk,l}{\bf x}_{lj}
}_{\text{Intracell Interference}}
\!+\!
\underbrace{
\sum_{i\neq l}{\bf H}_{lk,i}{\bf x}_{i}
}_{\text{Intercell Interference}}
\!+\!
\underbrace{{\bf  n}_{lk}
}_{\text{Noise}},\label{eq12}
\end{align}
where ${\bf n}_{lk}$ is the zero-mean additive white Gaussian noise at U$_{lk}$ with covariance matrix $\sigma^2{\bf I}$, where ${\bf I}$ denotes the identity matrix.
  {
Note that, in \eqref{eq12}, the differences between the intercell and intracell links are carefully taken into account, and each user is indeed affected by both \textit{intercell} and \textit{intracell} interference.} In this paper, we treat interference as noise, which is optimal for maximizing the sum rate \cite{annapureddy2009gaussian} or the generalized degree of freedom \cite{geng2015optimality} when the interference is weak. An alternative strategy for treating interference as noise (TIN) is to detect and cancel interference, which is known as successive interference cancellation (SIC), and it is optimal when the interference is strong. SIC requires more advanced user devices as well as more sophisticated signaling design. Moreover, to detect and cancel interference at the users takes some time, which may lead to violating the latency constraint in URLLC-related applications. 

 {Note that we employ Gaussian signaling in this work similar to most studies of wireless communication systems both with and without RISs. In practice, typically discrete constellations are employed. Studies based on Gaussian signaling are nevertheless important since they provide valuable insights into the system performance and represent an upper bound for the performance of the technologies studied. Additionally, there are studies on the comparison of Gaussian signals and discrete constellations, e.g., \cite{forney1998modulation, santamaria2018information, javed2020journey}. The performance gap between the discrete constellations and Gaussian signals grows as the number of bits/symbols in the discrete constellations increases, but it eventually saturates. To account for this performance gap and the idealized assumption, one can employ a signal-to-noise ratio (SNR) offset.}

 \subsection{Channel Dispersion, Rate and EE Expressions}
 A MIMO channel can be modeled as a set of parallel AWGN channels, and  \cite[Theorem 78]{polyanskiy2010} can be employed to obtain the achievable rate of MIMO channels associated with FBL coding. Note that \cite[Theorem 78]{polyanskiy2010} is based on the optimal power allocation for a point-to-point MIMO communication link; however, the FBL rate expressions can be formulated for any arbitrary power allocation as shown in \cite[Section 4.5.4]{polyanskiy2010}. In the following lemma, we calculate the achievable FBL rate of users, when the interference is treated as noise for decoding the corresponding signal at the receivers. 
\begin{lemma}[\!\cite{polyanskiy2010}]\label{lem-r}
The second-order rate of user U$_{lk}$ for  FBL coding along with the normal approximation (NA) is given by
\begin{equation}\label{1-multi}
r_{lk}=
\underbrace{\log \left|{\bf I} +{\bf D}^{-1}_{lk}{\bf S}_{lk} \right|}_{\text{\small Shannon Rate}}
-Q^{-1}(\epsilon)\sqrt{\frac{V_{lk}}{n_t}},
\end{equation}
where $n_t$ is the packet length, ${\bf S}_{lk}={\bf H}_{lk,l}{\bf P}_{lk}{\bf H}_{lk,l}^H$ is the covariance matrix of the desired signal at the user U$_{lk}$, while  ${\bf D}_{lk}$  is the covariance  matrix of the interfering signals plus noise, given by
\begin{equation}
{\bf D}_{lk}\!=\!\sigma^2{\bf I}
\!+\!\!\!
\sum_{i=1,i\neq l}^L\!\!{\bf H}_{lk,i}{\bf P}_i{\bf H}_{lk,i}^H
\!\!+\!\!\!
\sum_{j=1,j\neq k}^K{\bf H}_{lk,l}{\bf P}_{lj}{\bf H}_{lk,l}^H.
\end{equation}
Here, the first-order Shannon rate can also be written as
\begin{equation}\label{2}
C_{lk}\!=\!
{\log\! \left|{\bf I}\! +\!{\bf D}^{-1}_{lk}{\bf S}_{lk} \right|}\!\!
=\!\log \left|{\bf I} +{\bf \Lambda}_{lk} \right|\!=\!\!\!\sum_{i=1}^{I}\!\log\!\left(\! 1\!+\!\lambda_{lki} \right)\!,
\end{equation}
where ${\bf \Lambda}_{lk}=\text{diag}\left(\lambda_{lk1},\lambda_{lk2},\cdots,\lambda_{lkI}\right)$ is a diagonal matrix, containing the non-zero eigenvalues of the positive semidefinite (PSD) matrix ${\bf D}^{-1}_{lk}{\bf S}_{lk}$, and  
 {
$I\leq\min(N_{BS,l},N_{u,lk})$
} is equal to the rank of ${\bf H}_{lk,l}{\bf P}_{lk}{\bf H}_{lk,l}^H$, which also represents the number of parallel channels. The parameter $\lambda_{lki}$ is actually the signal-to-interference-plus-noise ratio (SINR) at the $i$-th parallel channel of user U$_{lk}$.
Finally, $V_{lk}=\sum_{i=1}^{I} V_{lki}$ is the channel dispersion of U$_{lk}$, where $V_{lki}$ is the channel dispersion of the $i$-th parallel channel of user U$_{lk}$. 
\end{lemma}  
The optimal channel dispersion of the $i$-th parallel channel is given by \cite{scarlett2016dispersion}
\begin{equation} 
V_{lki}=1-\frac{1}{(1+\lambda_{lki})^2},
\end{equation}
where is $\lambda_{lki}$ is given in Lemma \ref{lem-r}.
Unfortunately, the optimal channel dispersion attains the minimum value of $V_{lki}$ for all $l$, $k$, and $i$, but Gaussian signals cannot achieve it in the presence of interference. In \cite{scarlett2016dispersion}, a coding scheme was proposed for independent, identically distributed (iid) Gaussian signals in interference channels, which has the following channel dispersion
\begin{equation} 
V_{lki}=2\left(1-\frac{1}{1+\lambda_{lki}}\right).
\end{equation}
In the following lemma, we present closed-form matrix expressions for the optimal channel dispersion and the achievable channel dispersion in \eqref{1-multi}. 

\begin{lemma}\label{lem-1}
The optimal channel dispersion can be written as
\begin{equation} \label{dis-1}
V_{lk}=\text{\em Tr}\left({\bf I}-({\bf I}+{\bf D}^{-1}_{lk}{\bf S}_{lk})^{-2}\right).
\end{equation}
Additionally,
the achievable channel dispersion for the  scheme proposed in \cite{scarlett2016dispersion} can be written in the following matrix format 
\begin{align} \label{eq=10}
V_{lk}&=2\text{\em Tr}\left({\bf I}-({\bf I}+{\bf D}^{-1}_{lk}{\bf S}_{lk})^{-1}\right)
\\
\label{eq=11}
&=2\text{\em Tr}\left({\bf I}-{\bf D}_{lk}({\bf D}_{lk}+{\bf S}_{lk})^{-1}\right)
\\
\label{eq=12}
&=2\text{\em Tr}\left({\bf S}_{lk}({\bf D}_{lk}+{\bf S}_{lk})^{-1}\right).
\end{align}
\end{lemma} 
\begin{proof}
It is widely exploited that the trace of a positive semi-definite matrix is equal to the summation of its eigenvalues. It can be readily verified that the non-zero eigenvalues of $({\bf I}+{\bf D}^{-1}_{lk}{\bf S}_{lk})^{-2}$ are equal to $\left(1+\lambda_{lki}\right)^{-2}$, $i\in\{1,2,\cdots,I\}$,
which proves the equality in \eqref{dis-1}. Note that if ${\bf X}$ is a positive semi-definite matrix with non-zero eigenvalues $\lambda_i$, then its pseudo-inverse, denoted as ${\bf X}^{-1}$, is also a positive semi-definite matrix, and its non-zero  eigenvalues are $\lambda_i^{-1}$.

Similarly, it can be readily verified that the eigenvalues of $({\bf I}+{\bf D}^{-1}_{lk}{\bf S}_{lk})^{-1}$ are equal to $\left(1+\lambda_{lki}\right)^{-1}$, $i\in\{1,2,\cdots,I\}$,
which yields \eqref{eq=10}. Employing a simple matrix factorization, it can be easily verified that \eqref{eq=10}, \eqref{eq=11}, and \eqref{eq=12} are equivalent.
\end{proof} 
  {
\begin{remark}
   The reliability constraint is modeled by utilizing the maximum tolerable error rate, $\epsilon$, in the FBL rates. Moreover, the coding length $n_t$ should be proportional to the tolerable latency. Indeed, a more stringent latency constraint leads to employing a shorter block length \cite{xu2022rate, xu2022max}. Additionally, the latency constraint can be modeled as a constraint on the minimum rate as discussed in \cite[Remark 1]{liu2023energy}, \cite[Sec. II.D]{soleymani2023optimization}, \cite{he2021beamforming}. 
\end{remark}
} 
The EE of U$_{lk}$ is defined as \cite{zappone2015energy}
\begin{equation}
e_{lk}=\frac{r_{lk}}{P_c+\eta\text{Tr}\left(\mathbf{P}_{lk}\right)},
\end{equation}
where $\eta^{-1}$ is the power efficiency of the transmit devices at the BSs, $P_c$ is the constant power consumption of the system (including the devices of the BSs, RISs and U$_{lk}$) to transmit data to a user, which is given by \cite[Eq. (27)]{soleymani2022improper}, and
 $r_{lk}$ is given by Lemma \ref{lem-r}. Note that to compute $P_c$, the constant power of the devices of the BSs and RISs is normalized by the number of users served.
Moreover, the global EE (GEE) of the network is defined as \cite{zappone2015energy}
 \begin{equation}
g=\frac{\sum_{lk}r_{lk}}{LKP_c+\eta\sum_l\text{Tr}\left(\mathbf{P}_{l}\right)},
\end{equation}
which quantifies how energy efficient the network is. 
 {Finally, the transmission delay of U$_{lk}$ upon transmitting a packet with length $n_t$ is $d_{lk}=\frac{n_t}{r_{lk}}$.}

\subsection{Problem Statement}
We consider a general optimization problem for URLLC systems formulated as follows
\begin{subequations}\label{opt}
\begin{align}
 \underset{\{\mathbf{P}\}\in\mathcal{P},\{\bf{\Psi}\}\in\mathcal{T}
 }{\max}\,\,\,  & 
  f_0\!\left(\left\{\mathbf{P}\right\}\!,\!\{\bf{\Psi}\}\right)\!\!\!\! \\
  \text{s.t.}\,\,   \,&  f_i\left(\left\{\mathbf{P}\right\}\!,\!\{\bf{\Psi}\}\right)\geq0,\,\forall i,
\\
\label{4-c}
 &
r_{lk}\geq r^{th},\,\,\,\,\forall l,k,
 \end{align}
\end{subequations}
where constraint \eqref{4-c} can be interpreted as a latency constraint for each user, as discussed in \cite[Sec. II.D]{soleymani2023spectral}. {Moreover, functions $f_i$, $\forall i$ are, in general, non-linear functions of the optimization variables. These functions can be, but are not restricted to, linear functions of the rates/EEs and/or transmit/receive powers.  {For instance, $f_i$ can be a function of the sum rate, minimum rate/EE, transmission/receive power, interference temperature at a user, transmission delay and so on. Additionally, $f_i$ can be a non-linear function of the rates/EEs such as the geometric mean of the rates as in \cite{yu2021maximizing}. Note that \eqref{opt} may also include minimization problems such as the total power minimization subject to a given target rate, maximum transmission delay minimization, and interference temperature minimization. In this case, $f_i$ can be chosen as, e.g., $-\sum_l\text{Tr}({\bf P}_l)$ or $-\max_{\forall lk} \{d_{lk} \}$. }
Therefore, the general problem in \eqref{opt} can include an extensive range of optimization scenarios,
encompassing the maximization of the minimum weighted
rate, sum rate, global EE and minimum EE. We refer the reader
to  \cite[Sec. II.B]{soleymani2022rate} for more discussions on the format of the functions $f_i$s as well as of the family of optimization problems that can be formulated as \eqref{opt}.}
\section{Proposed optimization framework}\label{sec-iii}
In this section, we propose iterative schemes for solving 
\eqref{opt} 
by leveraging  AO, MM-based, and FP algorithms. 
Specifically, we first fix the RIS coefficients to $\{{\bf \Psi}^{(t-1)}\}$ and update the transmit covariance matrices as $\{{\bf P}^{(t)}\}$ by solving \eqref{opt}. We then alternate and update the RIS coefficients, while 
$\{{\bf P}\}$ is fixed to $\{{\bf P}^{(t)}\}$. We iterate this procedure until  convergence is reached. 
Unfortunately, the optimization problems are non-convex and complicated even when the RIS elements (or covariance matrices) are fixed. Thus, we propose a suboptimal scheme based on MM to solve the corresponding problems.
Below, we present our solutions for updating the transmit covariance matrices and RIS elements in separate subsections. 

\subsection{Updating Transmit Covariance Matrices}\label{sec-iii-a}
To update $\{{\bf P}\}$, 
we introduce a new set of variables $\{{\bf Q}\}=\{{\bf Q}_{lk}\}_{\forall lk}$, where ${\bf Q}_{lk}$ is a positive semi-definite matrix and ${\bf P}_{lk}={\bf Q}_{lk}{\bf Q}_{lk}^H$. Equivalently, we can compute ${\bf Q}_{lk}$ as ${\bf P}_{lk}^{1/2}$. 
To attain a suboptimal solution for \eqref{opt}, we leverage an MM-based technique. More specifically, we first obtain suitable concave surrogate functions for the rates. Then, we update ${\bf P}_{lk}^{(t)}={\bf Q}_{lk}^{(t)}{\bf Q}_{lk}^{(t)^H}$, $\forall l,k$, by solving the corresponding surrogate optimization problems. 
To derive concave lower bounds for the FBL rates, we utilize the bounds in the following lemmas.  
\begin{lemma}\label{lem-2}
Consider the arbitrary matrices ${\bf \Lambda}$, $\bar{\bf \Lambda}$ and positive semi-definite matrices  ${\bf \Upsilon}$, $\bar{\bf \Upsilon}$. Then, the following inequality holds for all feasible ${\bf \Lambda}$, $\bar{\bf \Lambda}$, ${\bf \Upsilon}$, and $\bar{\bf \Upsilon}$:
\begin{multline}
\label{eq10}
f\left({\bf \Lambda},{\bf \Upsilon}\right)=\text{\em Tr}\left({\bf \Upsilon}^{-1}{\bf \Lambda}{\bf \Lambda}^H\right)\geq 
2\mathfrak{R}\left\{\text{\em Tr}\left(\bar{\bf \Upsilon}^{-1}\bar{\bf \Lambda}{\bf \Lambda}^H\right)\right\}
\\-
\text{\em Tr}\left(\bar{\bf \Upsilon}^{-1}\bar{\bf \Lambda}\bar{\bf \Lambda}^H\bar{\bf \Upsilon}^{-1}{\bf \Upsilon}\right),
\end{multline}
where $\mathfrak{R}\{x\}$ returns the real value of $x$.
\end{lemma}  
\begin{proof} 
Function $f\left({\bf \Lambda},{\bf \Upsilon}\right)$ is jointly convex in ${\bf \Lambda}$ and ${\bf \Upsilon}$ \cite{sun2017majorization}. Thus, we can employ the first-order Taylor expansion to obtain an affine lower-bound for $f(\cdot)$ as follows
\begin{multline}
\label{eq11}
f\left({\bf \Lambda},{\bf \Upsilon}\right)\geq 
f\left(\bar{\bf \Lambda},\bar{\bf \Upsilon}\right)+2\mathfrak{R}\left\{\frac{\partial f\left({\bf \Lambda},{\bf \Upsilon}\right)}{\partial {\bf \Lambda}}\vert_{\bar{\bf \Lambda},\bar{\bf \Upsilon}}\left({\bf \Lambda}-\bar{\bf \Lambda}\right)\right.
\\
+\left.\frac{\partial f\left({\bf \Lambda},{\bf \Upsilon}\right)}{\partial {\bf \Upsilon}}\vert_{\bar{\bf \Lambda},\bar{\bf \Upsilon}}\left({\bf \Upsilon}-\bar{\bf \Upsilon}\right)\right\},
\end{multline}
where $\bar{\bf \Lambda}$ and $\bar{\bf \Upsilon}$ are any arbitrary feasible points, and 
$\frac{\partial f\left({\bf \Lambda},{\bf \Upsilon}\right)}{\partial {\bf \Lambda}}\vert_{\bar{\bf \Lambda},\bar{\bf \Upsilon}}$ (or $\frac{\partial f\left({\bf \Lambda},{\bf \Upsilon}\right)}{\partial {\bf \Upsilon}}\vert_{\bar{\bf \Lambda},\bar{\bf \Upsilon}}$) is the derivative of $f(\cdot)$ with respect to ${\bf \Lambda}$ (or ${\bf \Upsilon}$) at $\bar{\bf \Lambda}$ and $\bar{\bf \Upsilon}$.
Replacing the corresponding derivatives in \eqref{eq11} and simplifying the equation results in \eqref{eq10}.
\end{proof}
\begin{lemma}[\!\cite{soleymani2022improper}]\label{lem-3} 
Consider the arbitrary matrices ${\bf \Lambda}$ and $\bar{{\bf \Lambda}}$, and positive definite matrices ${\bf \Upsilon}$ and $\bar{{\bf \Upsilon}}$. Then, we have:
\begin{multline} 
\ln \left|\mathbf{I}+{\bf \Upsilon}^{-1}{\bf \Lambda}{\bf \Lambda}^H\right|\geq
 \ln \left|\mathbf{I}+{\bf \Upsilon}^{-1}\bar{{\bf \Lambda}}\bar{{\bf \Lambda}}^H\right|
\\-
\text{{\em Tr}}\left(
\bar{{\bf \Upsilon}}^{-1}
\bar{{\bf \Lambda}}\bar{{\bf \Lambda}}^H
\right)
+
2\mathfrak{R}\left\{\text{{\em Tr}}\left(
\bar{{\bf \Upsilon}}^{-1}
\bar{{\bf \Lambda}}{\bf \Lambda}^H
\right)\right\}\\
-
\text{{\em Tr}}\left(
(\bar{{\bf \Upsilon}}^{-1}-(\bar{{\bf \Lambda}}\bar{{\bf \Lambda}}^H + \bar{{\bf \Upsilon}})^{-1})^H({\bf \Lambda}{\bf \Lambda}^H+{\bf \Upsilon})
\right).
\label{lower-bound}
\end{multline}
\end{lemma}

Upon employing the concave lower bounds in Lemma \ref{lem-2} and Lemma \ref{lem-3}, we can obtain a concave lower bound for the FBL rates with the NA approximation as presented in the following lemma.
\begin{lemma}\label{lem-4}
A concave lower bound for $r_{lk}$ is given by 
\begin{multline}
\label{eq24}
r_{lk}\geq \tilde{r}_{lk}= a_{lk}
+2\sum_{ij}\mathfrak{R}\left\{\text{{\em Tr}}\left(
{\bf A}_{lk,ij}\mathbf{Q}_{ij}^H
\bar{\mathbf{H}}_{lk,i}^H\right)\right\}
\\
-
\text{{\em Tr}}\left(
{\bf B}_{lk}(\mathbf{H}_{lk,l}\mathbf{Q}_{lk}\mathbf{Q}_{lk}^H\mathbf{H}_{lk,l}^H+\mathbf{D}_{lk})
\right)
\end{multline}
where 
\begin{align*}
a_{lk}&\!=\!\ln\left|\!{\bf I}\!+\!\bar{\bf D}_{lk}^{-1}\bar{\bf S}_{lk}\!\right|\!\!
-\!\!
\text{{\em Tr}}\!\left(
\bar{\bf D}_{lk}^{-1}\bar{\bf S}_{lk}
\right)\!
\!-\!\frac{Q^{-1}\!(\epsilon)\!({\bar{V}_{lk}}\!\!+\!2I
)}{2\sqrt{n_t\bar{V}_{lk}}},
\\
{\bf A}_{lk,ij}\!&\!=\!\!\left\{\!\!\!\! \begin{array}{ll}
\bar{\mathbf{D}}^{-1}_{lk}\bar{\mathbf{H}}_{lk,l}\bar{\mathbf{Q}}_{lk}
&\!
\text{\em if}
\,\,i=l,
\,j=k,\\
\frac{Q^{-1}(\epsilon)}{\sqrt{n_t\bar{V}_{lk}}}
(\bar{\mathbf{S}}_{lk}+ \bar{\mathbf{D}}_{lk})^{-1}\bar{\mathbf{H}}_{lk,i}\bar{\mathbf{Q}}_{ij}
&
\!
\text{\em otherwise},
\end{array}\right.
\\
{\bf B}_{lk}&\!=\bar{\mathbf{D}}^{-1}_{lk}\!-\!(\bar{\mathbf{S}}_{lk}+ \bar{\mathbf{D}}_{lk})^{-1}\!\!
\\
&\hspace{.6cm}
+\!
\frac{Q^{-1}(\epsilon)}{\sqrt{n_t\bar{V}_{lk}}}
(\bar{\mathbf{S}}_{lk}+ \bar{\mathbf{D}}_{lk})^{-1}
\bar{\mathbf{D}}_{lk}(\bar{\mathbf{S}}_{lk}+ \bar{\mathbf{D}}_{lk})^{-1}\!,
\end{align*} 
where $\bar{\mathbf{D}}_{lk}$, 
$\bar{\mathbf{S}}_{lk}$, $\bar{\mathbf{Q}}_{ij}$, $\bar{V}_{lk}$, and $\bar{\mathbf{H}}_{lk,i}$, $\forall l,k,i,j$ are, respectively, the initial values of ${\mathbf{D}}_{lk}$, ${\mathbf{S}}_{lk}$, ${\mathbf{Q}}_{ij}$, $V_{lk}$, and ${\mathbf{H}}_{lk,i}$ at the current step, which are obtained upon replacing $\{\mathbf{P}\}$ by $\{\mathbf{P}^{(t-1)}\}$ and $\{\bf{\Psi}\}$ by $\{{\bf\Psi}^{(t-1)}\}$.
\end{lemma}  
\begin{proof}
Upon employing Lemma \ref{lem-3}, a concave lower bound can be obtained for the first-order Shannon rate  as
\begin{multline} \label{eq-22}
\left|{\bf I}+{\bf D}_{lk}^{-1}{\bf S}_{lk}\right|\geq \ln\left|{\bf I}+\bar{\bf D}_{lk}^{-1}\bar{\bf S}_{lk}\right|
-
\text{{Tr}}\left(
\bar{\bf D}_{lk}^{-1}\bar{\bf S}_{lk}
\right)
\\
+
2\mathfrak{R}\left\{\text{{ Tr}}\left(
\bar{\mathbf{Q}}^H_{lk}\bar{\mathbf{D}}^{-1}_{lk}\mathbf{Q}_{lk}
\right)\right\}
\\
-\!
\text{{ Tr}}\!\left(\!
(\bar{\mathbf{D}}^{-1}_{lk}\!-(\bar{\mathbf{S}}_{lk}\!+\! \bar{\mathbf{D}}_{lk})^{-1})^H(\bar{\bf H}_{lk,l}\mathbf{Q}_{lk}\mathbf{Q}^H_{lk}\bar{\bf H}_{lk,l}^H\!+\!\mathbf{D}_{lk})\!
\right)\!.
\end{multline}
Next, we obtain a concave lower bound for $-Q^{-1}(\epsilon)\sqrt{\frac{V_{lk}}{n_t}}$, which is equivalent to obtaining a convex upper bound for $\sqrt{V_{lk}}$. 
To this end, we first employ the following inequality 
\begin{equation}
\sqrt{V_{lk}}\leq \frac{\sqrt{\bar{V}_{lk}}}{2}+\frac{V_{lk}}{2\sqrt{\bar{V}_{lk}}},
\end{equation}
which is non-convex since $V_{lk}$ is not convex in $\{{\bf Q}\}$.
Upon employing Lemma \ref{lem-2}, a convex upper bound for $V_{lk}(\cdot)$ can be obtained as 
\begin{multline} \label{eq-24}
V_{lk}\leq 2\text{{Tr}}({\bf I})-
4
\sum_{[ij]\neq[lk]}\mathfrak{R}\left\{\text{{Tr}}\left(
{\bf A}_{lk,ij}\mathbf{Q}_{ij}^H
\bar{\mathbf{H}}_{lk,i}^H\right)\right\}
\\
+
2\text{{Tr}}\left(
((\bar{\mathbf{S}}_{lk}+ \bar{\mathbf{D}}_{lk})^{-1})^{-1}
\bar{\mathbf{D}}_{lk}(\bar{\mathbf{S}}_{lk}+ \bar{\mathbf{D}}_{lk})^{-1})^{-1}\right.
\\
\times\left.
(
\mathbf{H}_{lk,l}\mathbf{Q}_{lk}(\mathbf{H}_{lk,l}\mathbf{Q}_{lk})^H\!\!\!+\!\mathbf{D}_{lk})
\right),
\end{multline}
where $[ij]\neq[lk]$ includes all possible $i$, $j$ pairs, except for the case where $i=l$ and simultaneously $j=k$.
Substituting the concave lower bound in \eqref{eq-22} and the convex upper bound in \eqref{eq-24} into the FBL rate expression proves the lemma. 
\end{proof}
\begin{remark}
The concave lower bound in Lemma \ref{lem-3} is quadratic in $\{{\bf Q}\}$, and consists of a constant term, a linear/affine term, and a quadratic term.
\end{remark} 
We denote the surrogate functions for $f_i$ by $\tilde{f}_i$, which are obtained by substituting the concave lower bounds $\tilde{r}_{lk}$ in \eqref{opt}.
For instance, if ${f}_i$ is equal to the sum rate, then $\tilde{f}_i=\sum_{\forall lk} \tilde{r}_{lk}$.
Moreover, if ${f}_i$ represents the EE of U$_{lk}$, then we have 
\begin{equation}\label{eq31}
\tilde{f}_i=\tilde{e}_{lk}=\frac{\tilde{r}_{lk}}{P_c+\eta\text{Tr}\left(\mathbf{Q}_{lk}\mathbf{Q}_{lk}^H\right)},
\end{equation}
which is a fractional function of $\{{\bf Q}\}$ with a concave numerator and convex denominator. 
 {Moreover, if ${f}_i$ is a function of the transmission delay, then we have $\tilde{f}_i=-\frac{n_t}{\tilde{r}_{lk}}$.}
Note that although the surrogate lower bounds for the rates in Lemma \ref{lem-4} are concave, the $\tilde{f}_i$s are not necessarily concave, since they might be a linear function of the EE metrics. Substituting the $f_i$s by the $\tilde{f}_i$s leads to 
\begin{subequations}\label{opt-sur}
\begin{align}
 \underset{\{\mathbf{Q}\}
 }{\max}\, \,\, & 
  \tilde{f}_0\!\left(\left\{\mathbf{Q}\right\}\!,\!\{{\bf\Psi}^{(t-1)}\}\right)\!\!\!\! \\
  \text{s.t.} \,\,  \,&  \tilde{f}_i\left(\left\{\mathbf{Q}\right\}\!,\!\{{\bf\Psi}^{(t-1)}\}\right)\geq0,\,\forall i,
\\
\label{32-c}
 &
\tilde{r}_{lk}\geq r^{th},\,\,\,\,\forall l,k,
\\
\label{32-d}
&
\sum_{k}\text{Tr}\left(\mathbf{Q}_{lk}\mathbf{Q}_{lk}^H\right)\leq p_l, \forall l.
 \end{align}
\end{subequations}
The optimization problem \eqref{opt-sur} is convex for the maximization of the minimum and/or sum rates.
Hence, it can be efficiently solved by existing numerical tools. 
 {Note that our framework
can also optimize other SE metrics, such as the geometric
mean of users. As shown in \cite{yu2021maximizing}, the maximization of the geometric mean of the users can be solved by solving a sequence of weighted sum rate maximization problems, which can be efficiently handled by our framework.} 

Unfortunately, \eqref{opt-sur} is non-convex for GEE maximization as well as for the maximization of the minimum weighted EE of the users, since the EE and/or GEE functions are not concave in $\{{\bf Q}\}$. Fortunately, a solution of the minimum weighted EE of the users and/or GEE maximization problems can be obtained by Dinkelbach-based algorithms, since the numerator of $\tilde{e}_{lk}$ (or $\tilde{g}$) is concave, while its denominator is  convex.  {The problem in \eqref{opt-sur} is not convex, when minimizing the maximum transmission delay of the users, which can be considered as a latency metric from a physical layer point of view. In the following, we solve \eqref{opt-sur} for the maximization of the minimum weighted EE of the users, the maximization of the GEE and the minimization of the maximum transmission delay.} 
\subsubsection{Maximization of the Minimum EE}
In this case, \eqref{opt-sur} can be written as  
\begin{subequations}\label{opt-sur2}
\begin{align}
 \underset{\{\mathbf{Q}\},e
 }{\max}\, \,\, & 
 e, 
&
  \text{s.t.} \,\,  \,&  \tilde{e}_{lk}=\frac{\tilde{r}_{lk}}{P_c+\eta\text{Tr}\left(\mathbf{Q}_{lk}\mathbf{Q}_{lk}^H\right)}\geq e,\,\forall i,
\\
&&&\eqref{32-c},\eqref{32-d}.
 \end{align}
\end{subequations}
Upon employing the generalized Dinkelbach algorithm (GDA), we can derive the globally optimal solution of \eqref{opt-sur2} by iteratively solving the convex optimization problem \cite{zappone2015energy}
\begin{subequations}\label{opt-sur3}
\begin{align}
 \underset{\{\mathbf{Q}\},e
 }{\max} \,\, & 
 e,\!\!\!\! 
&
  \text{s.t.} \,\,  &  {\tilde{r}_{lk}}\!-\!\mu^{(n)}\left(P_c\!+\!\eta\text{Tr}\left(\mathbf{Q}_{lk}\mathbf{Q}_{lk}^H\right)\right)\!\geq e,\,\forall i,
\\
&&&\eqref{32-c},\eqref{32-d},
 \end{align}
\end{subequations}
and updating $\mu^{(n)}$ as 
\begin{equation}\label{mu}
\mu^{(n)}\!\!=\!\min_{lk}\!\!\left\{\!\tilde{e}^{(n-1)}_{lk}\!\right\}\!\!=\!\min_{lk}\!\!\left\{\!\!\frac{\tilde{r}_{lk}\left(\mathbf{Q}_{lk}^{(n-1)}\right)}{P_c\!+\!\eta\text{Tr}\left(
\!\!\mathbf{Q}_{lk}^{(n-1)}\!\mathbf{Q}_{lk}^{(n-1)^H}\!
\right)}\!\!\right\}\!.
\end{equation}
Note that $n$ is the number of iterations in the inner loop, i.e., the number of GDA iterations. 

\subsubsection{Maximization of the GEE} 
In this case, \eqref{opt-sur} is equivalent to 
\begin{align}
\label{opt-sur-gee}
 \underset{\{\mathbf{Q}\}
 }{\max}\, \,\, & 
\frac{\sum_{l,k}\tilde{r}_{lk}}{\sum_{l,k}\left(P_c+\eta\text{Tr}\left(\mathbf{Q}_{lk}\mathbf{Q}_{lk}^H\right)\right)}
&
  \text{s.t.}\,\,
 &\eqref{32-c},\eqref{32-d}.
\end{align}
Employing the Dinkelbach algorithm, a globally optimal solution of \eqref{opt-sur-gee} can be found by iteratively solving 
\cite{zappone2015energy}
\begin{align*}
 \underset{\{\mathbf{Q}\}
 }{\max}\, & 
 \sum_{l,k}\!\tilde{r}_{lk}\!-\!\mu^{(n)}
\sum_{l,k}\left(P_c+\eta\text{Tr}\left(\mathbf{Q}_{lk}\mathbf{Q}_{lk}^H\right)
\right) \!\!
&
  \text{s.t.}  \,&\eqref{32-c},\eqref{32-d}, 
 \end{align*}
and updating $\mu^{(n)}$ as 
\begin{equation*}
\mu^{(n)}\!\!=\tilde{g}^{(n-1)}_{lk}=\!\frac{\sum_{l,k}\tilde{r}_{lk}\left(\mathbf{Q}_{lk}^{(n-1)}\right)}{
\sum_{l,k}\!\left(\!P_c\!+\!\eta\text{Tr}\left(\mathbf{Q}_{lk}^{(n-1)}\mathbf{Q}_{lk}^{(n-1)^H}\right)
\!\right)}\!.
\end{equation*} 
\subsubsection{ {Minimization of the maximum transmission delay}}
 {To minimize the maximum delay, we have to solve
\begin{align}
\label{opt-sur-delay}
 \underset{\{\mathbf{Q}\}
 }{\min}\,  & 
\max_{\forall lk} \left\{\frac{n_t}{\tilde{r}_{lk}}\right\}
&
  \text{s.t.}\,\,
 &\eqref{32-c},\eqref{32-d},
\end{align}
which is equivalent to maximizing the minimum rate of users as
\begin{align}
\label{opt-sur-delay2}
 \underset{\{\mathbf{Q}\}
 }{\max}\,  & 
\min_{\forall lk} \left\{\tilde{r}_{lk}\right\}
&
  \text{s.t.}\,\,
 &\eqref{32-c},\eqref{32-d},
\end{align}
which is a convex optimization problem.
Note that the transmission delay $d_{lk}$ is a monotonically decreasing function of $r_{lk}$, and it is minimized when $r_{lk}$ is maximized, which makes the solution of \eqref{opt-sur-delay} equivalent to \eqref{opt-sur-delay2}.
}
\subsubsection{Discussion on Single-stream Data Transmission}
 {In this case, the rank of  matrix ${\bf P}_{lk}$ for all $l,k$ is equal to one. This means that matrix ${\bf P}_{lk}$ can be written as ${\bf P}_{lk}={\bf q}_{lk}{\bf q}_{lk}^H$, where ${\bf q}_{lk}\in\mathbb{C}^{ N_{BS,l}\times 1}$.
In other words, when single-stream data transmission is employed, the BSs perform only beamforming to transmit data, and we have to optimize the beamforming vectors instead of transmit covariance matrices. Therefore, the computational complexity of  single-stream data transmission is lower than  that of employing multiple streams. 
However, this lower computational complexity is attained at the cost of a significant performance loss, especially when the maximum number of streams, i.e., {$\min(N_{BS,l},N_{u,lk})$} for U$_{lk}$, increases. Indeed, as the network size increases, more advanced transmission and resource allocation techniques, involving higher complexities, are needed to provide  satisfactory performance.}

 {
Note that if the transmitter and/or the receiver are equipped with only a single antenna, we are restricted to single-stream data transmissions. Obviously, the single-stream scheme becomes increasingly suboptimal, when the maximum number of streams grows. To derive single-stream data transmission solutions, we only have to replace matrices ${\bf Q}_{lk}$ by  vectors ${\bf q}_{lk}$ and employ the lower bounds in Lemma \ref{lem-4}. Indeed, the schemes proposed in this subsection can be applied for both single- and multiple-stream data transmission.  
 }
\subsection{Optimizing the RIS Elements}
Now, we update $\{\bf{\Psi}\}$ by solving \eqref{opt} for fixed
$\{{\bf P}^{(t)}\}$, i.e.,  
\begin{subequations}\label{opt-t}
\begin{align}
 \underset{\{\bf{\Psi}\}\in\mathcal{T}
 }{\max}\,\,\,  & 
  f_0\!\left(\left\{\mathbf{P}^{(t)}\right\}\!,\!\{\bf{\Psi}\}\right)\!\!\!\! \\
  \text{s.t.}\,\,   \,&  f_i\left(\left\{\mathbf{P}^{(t)}\right\}\!,\!\{\bf{\Psi}\}\right)\geq0,\,\forall i,
\\
 &
r_{lk}\geq r^{th},\,\,\,\,\forall l,k,
 \end{align}
\end{subequations}
which is non-convex since the rates are not concave in $\{\bf{\Psi}\}$ and the set $\mathcal{T}$ might be non-convex. 
 {Note that as discussed in Section \ref{sec-iii-a}, the minimization of the maximum transmission delay is equivalent to the maximization of the minimum rate. Hence, our focus in this subsection is on the SE and EE metrics.} In the following, we first consider reflective RISs and then, describe how the solution can be applied to the more sophisticated STAR-RIS, utilizing the MS scheme.  
\subsubsection{ {Reflective RISs}}
To find a suboptimal solution for \eqref{opt-t}, we leverage an approach based on MM. Specifically, we first obtain a suitable concave lower bound for the rates and then, convexify  $\mathcal{T}$ if it is not already a convex set. 
Since the rates have similar structures in $\{{\bf Q}\}$ and in $\{\bf{\Psi}\}$, we can utilize the concave lower bounds in Lemma \ref{lem-4} to attain concave lower bounds for the rates and construct suitable surrogate optimization problems for updating  $\{\bf{\Psi}\}$.
In the subsequent corollary, we provide the concave lower bounds for the rates.The proof of this corollary closely resembles that of Lemma \ref{lem-4}, and thus, we omit it here. 

\begin{corollary}\label{cor-1}
A concave lower bound for $r_{lk}$ is given by 
\begin{multline}
\label{eq34}
r_{lk}\geq \hat{r}_{lk}= a_{lk}
+2\sum_{ij}\mathfrak{R}\left\{\text{{\em Tr}}\left(
{\bf A}_{lk,ij}\mathbf{Q}_{ij}^{(t)^H}
{\mathbf{H}}_{lk,i}^H\right)\right\}
\\
-
\text{{\em Tr}}\left(
{\bf B}_{lk}(\mathbf{H}_{lk,l}\mathbf{Q}_{lk}^{(t)}(\mathbf{H}_{lk,l}\mathbf{Q}_{lk}^{(t)})^H+\mathbf{D}_{lk})
\right),
\end{multline}
where the constant parameters $a_{lk}$, ${\bf A}_{lk,ij}$, and ${\bf B}_{lk}$ are defined as in Lemma \ref{lem-4}.
\end{corollary}

Substituting the surrogate functions for the rates, i.e., the $\hat{r}_{lk}$s, in \eqref{opt-t} yields
\begin{subequations}\label{opt-t-sur}
\begin{align}
 \underset{\{\bf{\Psi}\}\in\mathcal{T}
 }{\max}\,\,\,  & 
  \hat{f}_0\!\left(\left\{\mathbf{P}^{(t)}\right\}\!,\!\{\bf{\Psi}\}\right)\!\!\!\! \\
  \text{s.t.}\,\,   \,&  \hat{f}_i\left(\left\{\mathbf{P}^{(t)}\right\}\!,\!\{\bf{\Psi}\}\right)\geq0,\,\forall i,
\label{40-b}
\\
\label{40-c}
 &
\hat{r}_{lk}\left(\left\{\mathbf{P}^{(t)}\right\}\!,\!\{\bf{\Psi}\}\right)\geq r^{th},\,\,\,\,\forall l,k,
 \end{align}
\end{subequations}
which is convex if $\mathcal{T}$ is a convex set, i.e. when $\mathcal{T}_U$ is considered. For $\mathcal{T}_U$, the proposed scheme achieves convergence to a stationary point of \eqref{opt}. Note that the surrogate functions $\hat{f}_i\left(\left\{\mathbf{P}^{(t)}\right\}\!,\!\{\bf{\Psi}\}\right)$ are concave in $\{{\bf\Psi}\}$ even if they are linear functions of the EE metrics. The reason is that the powers (transmit covariance matrices) are fixed, and thus, the EE metrics are not fractional functions of $\{{\bf\Psi}\}$. 

Now, we convexify $\mathcal{T}_I$. The unit modulus constraint $|\psi_{m_n}|=1$ is equivalent to
\begin{align}
\label{38}
|\psi_{m_n}|^2&\leq 1,\\
|\psi_{m_n}|^2&\geq 1.
\label{39}
\end{align} 
The constraint $|\psi_{m_n}|^2\leq 1$ is convex; however, $|\psi_{m_n}|^2\geq 1$ is not, which makes \eqref{opt-t-sur} a non-convex problem.
Thus, we have to approximate \eqref{39} with a convex constraint to make \eqref{opt-t-sur} convex. 
To this end, we employ the convex-concave procedure (CCP) and rewrite \eqref{39} as
\begin{align}\label{40}
|\psi_{m_n}|^2\geq|\psi_{m_n}^{(t-1)}|^2-2\mathfrak{R}\{\psi_{m_n}^{(t-1)^*}(\psi_{m_n}-\psi_{m_n}^{(t-1)})\}\geq 1.
\end{align} 
To avoid potential numerical errors and speed up the convergence, we relax \eqref{40} as
\begin{equation}\label{41}
|\psi_{m_n}^{(t-1)}|^2-2\mathfrak{R}\{\psi_{m_n}^{(t-1)^*}(\psi_{m_n}-\psi_{m_n}^{(t-1)})\}\geq 1-\delta,
\end{equation}
where $\delta>0$. Now, we can approximate \eqref{opt-t-sur} as 
\begin{subequations}\label{42}
\begin{align}
 \underset{\{\bf{\Psi}\}
 }{\max}\,  & 
  \hat{f}_0\!\left(\left\{\mathbf{P}^{(t)}\right\}\!,\!\{\bf{\Psi}\}\right)\!\!\!\!\!\!\! &
  \text{s.t.} \,\,\,\,\,&  \eqref{40-b},\eqref{40-c},
\\
 &&&
\eqref{38},\eqref{41}, \,\, \forall m,n,
 \end{align}
\end{subequations}
which is convex. 
We denote the solution of \eqref{42} as $\{{\bf \Psi}^{(\star)}\}=\{{\bf \Psi}^{(\star)}_1,{\bf \Psi}^{(\star)}_2,\cdots,{\bf \Psi}^{(\star)}_M\}$. Because of the relaxation in \eqref{41}, it may happen that $\psi_{m_n}^{(\star)}$, i.e., the $n$-th coefficient of the diagonal matrix ${\bf \Psi}^{(\star)}_m$, does not satisfy $|\psi_{m_n}|=1$. Therefore, we normalize $\{{\bf \Psi}^{(\star)}\}$ as
\begin{equation}
\hat{\psi}_{m_n}=\frac{{\psi}_{m_n}^{(\star)}}{|{\psi}_{m_n}^{(\star)}|}, \hspace{1cm}\forall m,n.
\end{equation}
To guarantee the convergence, we update $\{{\bf \Psi}\}$ as 
 \begin{equation}\label{eq-42}
\{{\bf \Psi}^{(t)}\}=
\left\{
\begin{array}{lcl}
\{\hat{{\bf\Psi}}\}&\text{if}&
f_0\left(\left\{\mathbf{P}^{(t)}\right\},\{\hat{{\bf\Psi}}\}\right)\geq
\\
&&
f_0\left(\left\{\mathbf{P}^{(t)}\right\},\{{\bf\Psi}^{(t-1)}\}\right)
\\
\{{\bf\Psi}^{(t-1)}\}&&\text{otherwise}.
\end{array}
\right.
\end{equation}
For $\mathcal{T}=\mathcal{T}_I$,  our proposed framework  converges since a non-decreasing sequence of objective functions (OF) $f_0$ is generated. 
For $\mathcal{T}=\mathcal{T}_U$, our framework converges to a stationary point of \eqref{opt} because $\mathcal{T}_U$ is convex. 
We summarize our algorithm for maximizing the minimum EE with $\mathcal{T}_U$ in Algorithm I.

\doublespacing 
\begin{table}[htb]
\small
\begin{tabular}{l}
\hline 
 {\textbf{Algorithm I} Maximization of the minimum EE for $\mathcal{T}_U$.} \\
\hline
\hspace{0.2cm}{\textbf{Initialization}}\\
\hspace{0.2cm}{Set $\gamma_1$, $\gamma_2$, $t=1$,  $\{\mathbf{P}\}=\{\mathbf{P}^{(0)}\}$, and$\{{\bf \Psi}\}=\{{\bf \Psi}^{(0)}\}$ }\\
\hline 
\hspace{0.2cm}
{\textbf{While} $\left(\underset{\forall lk}{\min}\,e^{(t)}_{lk}-\underset{\forall lk}{\min}\,e^{(t-1)}_{lk}\right)/\underset{\forall lk}{\min}\,e^{(t-1)}_{lk}\geq\gamma_1$ }\\ 
\hspace{.6cm}{{\bf Optimizing over} $\{\mathbf{P}\}$ {\bf by fixing} $\{{\bf \Psi}^{(t-1)}\}$}\\
\hspace{1.2cm}{Derive $\tilde{r}_{lk}$ 
according to Lemma \ref{lem-4}}\\ 
\hspace{1.2cm}{Derive $\tilde{e}_{lk}$ 
based on \eqref{eq31}}\\ 
\hspace{1.2cm}{Compute $\{{\bf Q}\}$ by solving \eqref{opt-sur2}, i.e., by running}\\
\hspace{1.2cm}{\textbf{While} $\left(\underset{\forall lk}{\min}\,\tilde{e}^{(n)}_{lk}-\underset{\forall lk}{\min}\,\tilde{e}^{(n-1)}_{lk}\right)/\underset{\forall lk}{\min}\,\tilde{e}^{(n-1)}_{lk}\geq\gamma_2$ }\\ 
\hspace{1.8cm}{Update $\mu^{(n)}$ based on \eqref{mu}}\\
\hspace{1.8cm}{Update $\{{\bf Q}\}$ by solving \eqref{opt-sur3}}\\
\hspace{1.2cm}{Compute $\{\mathbf{P}^{(t)}\}$ as $\mathbf{P}^{(t)}_{lk}=\mathbf{Q}^{(t)}_{lk}\mathbf{Q}^{(t)^H}_{lk}\,\, \forall lk$}\\
\hspace{.6cm}{{\bf Optimizing over} $\{{\bf\Psi}\}$ {\bf by fixing} $\{\mathbf{P}^{(t-1)}\}$}\\
\hspace{1.2cm}{Derive $\hat{r}_{lk}^{(t)}$  
according to Corollary 1}\\
\hspace{1.2cm}{Calculate $\{{\bf \Psi}^{(t)}\}$ by solving \eqref{opt-t-sur}}\\
\hspace{.6cm}{$t=t+1$}\\
\hspace{0.2cm}{\textbf{End (While)}}\\
\hspace{0.2cm}{{\bf Return} $\{\mathbf{P}^{(\star)}\}$ and $\{{\bf \Psi}^{(\star)}\}$.}\\
\hline 
\end{tabular} 
\end{table}
\singlespacing


\subsubsection{STAR-RIS using the MS scheme} The solution of \eqref{opt-t} for STAR-RIS with the MS scheme is very similar to the solution for a reflective RIS. 
In this case, we can still use the surrogate functions in Corollary \ref{cor-1} to make the rates a jointly concave function of the STAR-RIS coefficients, i.e., $\{{\bf\Psi}^{t}\}$ and $\{{\bf\Psi}^{r}\}$. For $\mathcal{T}_{SU}$, we can update the STAR-RIS parameters by solving \eqref{opt-t-sur}, and the  algorithm obtains a stationary point of \eqref{opt}. 
For $\mathcal{T}_{SI}$, we have to ``convexify'' the constraint in \eqref{eq-8}, which can be done by rewriting it as the two convex constraints in \cite[Eq. (34)]{soleymani2022rate} and \cite[Eq. (36)]{soleymani2022rate}. Then we can update the STAR-RIS coefficients similar to the proposed scheme for the reflective RIS. 
\subsection{Computational Complexity Analysis} 
Our optimization framework operates iteratively, with the actual computational complexity and runtime contingent upon the specific implementation of the algorithms. In this subsection, we calculate an approximate upper bound for the number of multiplications imposed by running our algorithms. To this end, we consider the maximization of the minimum rate with the feasibility set $\mathcal{T}_U$. The computational complexity of other optimization problems, including the weighted sum rate, the minimum EE, and global EE, can be similarly computed.  

Each iteration of our proposed framework consists of two steps. In the first step, 
 we optimize the transmit covariance matrices by solving \eqref{opt-sur}, which is convex when the minimum rates of the users are maximized. 
We solve the convex problem in \eqref{opt-sur} by numerical optimization tools. 
 {To numerically solve a convex optimization problem,  the number of Newton iterations increases proportionally with the square root of the number of its constraints \cite[Chapter 11]{boyd2004convex}, which is equal to $\sum_l(K_l+1)$ in \eqref{opt-sur} for the maximization of the minimum rate. Note that the maximization of the minimum rate can be written as in \cite[Eq. (30)]{soleymani2023spectral} and thus, has $\sum_l(K_l+1)$ constraints, considering the power budget in \eqref{32-d}. Now, we provide an approximate upper bound for the number of multiplications to find a solution in each Newton iteration. To solve each Newton iteration, $\sum_lK_l$ surrogate functions have to be computed  for the rates, $\tilde{r}_{lk}$. The surrogate rates $\tilde{r}_{lk}$ in Lemma \ref{lem-4} are quadratic in $\{{\bf Q}\}$, and the computational complexity to compute each  surrogate rate can be approximated as $\mathcal{O}\left[\sum_l\sum_kN_{BS,l}^2(2N_{BS,l}+N_{u,lk})\right]$. Note that the coefficients in \eqref{eq24} can be computed once at the beginning of the Newton iterations, and there is no need to recompute them in each Newton iteration to reduce the overall computational complexity of the framework. Finally, the  computational complexity to update $\{{\bf P}\}$ can be approximated as $\mathcal{O}\left[\mathcal{L}\sum_l\sum_kN_{BS,l}^2\sqrt{\sum_l(K_l+1)}(2N_{BS,l}+N_{u,lk})\right]$, where $\mathcal{L}=\sum_l K_l$ is the total number of users in the system.}

Now, we derive an approximation for the number of the multiplications needed to update $\{{\bf \Psi}\}$. To this end, we have to calculate the number of multiplications for solving the surrogate optimization problem in \eqref{opt-t-sur}, which is convex for the feasibility set $\mathcal{T}_U$. {The number of constraints in \eqref{opt-t-sur} is equal to $\sum_l K_l+\sum_mN_{RIS,m}$. Thus, the number of the Newton iterations grows with $\sqrt{\sum_l K_l+\sum_mN_{RIS,m}}$. To solve each Newton iteration, we have to compute $\mathcal{L}$ surrogate functions for the rates, $\hat{r}_{lk}$, as well as $L\mathcal{L}$ equivalent channels, according to \eqref{ch-equ}. To compute each channel, ${\bf H}_{lk,i}$, $\forall l,k,i$, there are approximately $\sum_mN_{u,lk}N_{BS,i}N_{RIS,m}$ multiplications, since the matrices ${\bf \Psi}_m$ are diagonal, which reduces the computational complexity. Moreover, the structure of the rates in \eqref{eq34} is very similar to the rates in \eqref{eq24}. Hence, the computational complexity of calculating $\hat{r}_{lk}$ in \eqref{eq34}  is on the same order of the computational complexity of attaining $\tilde{r}_{lk}$ as in \eqref{eq24} and can be approximated as $\mathcal{O}\left[\sum_l K_l N_{BS,l}^2(2N_{BS,l}+N_{u,lk})\right]$. Finally, the computational complexity of updating $\{{\bf \Psi}\}$ can be approximated as $\mathcal{O}[\sum_l\mathcal{L}\sqrt{\sum_l K_l+\sum_mN_{RIS,m}}( K_l N_{BS,l}^2(2N_{BS,l}+N_{u,lk})+\sum_mN_{u,lk}N_{BS,i}N_{RIS,m})]$.}
Assuming that the maximum number of iterations is equal to $T$, the computational complexity of solving the maximization of the minimum rate for $\mathcal{T}_U$ using our framework is $T$ times the summation of the computational complexities of updating $\{{\bf P}\}$ and $\{{\bf \Psi}\}$. 

\subsection{Discussion on Extending the Framework to Uplink} 
The framework can also be applied to the SE and EE maximization of the uplink along with FBL coding, since the structure of the rates with respect to the beamforming matrices and channels is very similar to the DL scenario considered in the paper. The detailed solution for the UL scenario is beyond the scope of this work. However, we provide some insights on how the proposed solutions can be modified to maximize the sum rate in UL communications in a multi-cell multiple-access channel (MAC), while the intercell interference is treated as noise. In this case, the sum rate of users associated with BS $l$ is
  \begin{equation}\label{(47)}
      r_l=\log|{\bf I}+{\bf D}_l^{-1}{\bf S}_l|-\frac{Q^{-1}(\epsilon)}{n_t}\sqrt{2\text{Tr}({\bf I}-{\bf D}_l({\bf D}_l+{\bf S}_l)^{-1})},
  \end{equation}
  where ${\bf S}_l=\sum_k{\bf H}_{lk,l}{\bf Q}_{lk}({\bf H}_{lk,l}{\bf Q}_{lk})^H$ is the covariance matrix of the signals decoded at BS $l$, and ${\bf D}_l=\sigma^2{\bf I}+\sum_{\forall ik,i\neq l}{\bf H}_{lk,l}{\bf Q}_{lk}({\bf H}_{lk,l}{\bf Q}_{lk})^H$ is the covariance matrix of the noise plus interference, where ${\bf Q}_{ik}$ is the beamforming matrix at U$_{ik}$, and ${\bf H}_{ik,l}$ is the uplink channel between U$_{ik}$ and BS $l$. As it can be easily verified, ${\bf S}_l$ and ${\bf D}_l$
in \eqref{(47)} are quadratic in channels and beamforming matrices, which follow a similar structure as the matrices ${\bf S}_{lk}$ and ${\bf D}_{lk}$ in \eqref{1-multi}. Thus, to obtain a quadratic and concave surrogate function for $r_l$, we can employ the bounds in Lemma \ref{lem-2} and Lemma \ref{lem-3}, and follow the steps in the proof of Lemma \ref{lem-4}. Once the surrogate functions for the rates are calculated, the beamforming matrices and RIS coefficients can be updated by solving the corresponding optimization problems according to our proposed framework.
 
\section{Numerical results}\label{sec-iv}
 In this section, we provide numerical results based on Monte Carlo simulations. To this end, we consider a two-cell system with one RIS and $K$ users per cell, similar to \cite[Fig. 2]{soleymani2022rate}. We also consider that each BS/user/RIS has $N_{BS}/N_u/N_{RIS}$ TAs/RAs/elements. Moreover, the locations and heights of the users/BSs/RISs are chosen similar to \cite{soleymani2022rate}.  We assume that the power budgets of the BSs are equal to $P$. 
 {To generate the channels, we assume that the links with respect to the RISs benefit from a line of sight (LoS) for both BSs and the users that are located in the same cell as the RIS. Therefore, these links follow Rician fading associated with a Rician factor of $3$. More particularly, these channels are generated according to \cite[(60)-(60)]{soleymani2022improper}. On the other hand, we assume that the direct links between the BSs and the users as well as the links through the RISs across the cells are of a non-LoS (NLoS) nature and consequently, follow a Rayleigh distribution. This means that each entry of the corresponding channel matrices follows a complex-valued proper Gaussian distribution with zero mean and unit variance. The large-scale fading of the links is modeled according to \cite[(59)]{soleymani2022improper}. The path-loss exponents of the LoS and NLoS links are $2.2$ and $3.75$, respectively.  
The noise power density and channel bandwidth are, respectively, assumed to be $-174$ dBm per Hz and $1.5$ MHz. The other simulation parameters are chosen according to \cite{soleymani2022improper, soleymani2022rate}.
}

In the following, we provide numerical results for SE and EE maximization in Section \ref{sec-iv-a} and Section \ref{sec-iv-b}, respectively. 
The  schemes considered in this section are as follows:
\begin{itemize}
\item {\bf RIS} (or {\bf RIS$_{I}$}): Our algorithms for MIMO RIS-assisted URLLC systems with multiple data streams per users, and $\mathcal{T}_U$ (or $\mathcal{T}_I$).

\item {\bf No-RIS}: The scheme for MIMO URLLC systems with multiple data streams per users, but without RIS.

\item {\bf RIS-Rand} (or {\bf S-RIS-Rand}): The algorithm for MIMO RIS-aided (or STAR-RIS-aided) URLLC systems with multiple data streams per users, but without optimizing RIS elements.

\item {\bf STAR-RIS}: Our algorithms for MIMO STAR-RIS-aided URLLC BCs with multiple data streams per users, $\mathcal{T}_U$, and the MS scheme.

\item {\bf SS-RIS}: Our algorithms for MIMO RIS-assisted URLLC systems with single-stream data transmission per users, and $\mathcal{T}_U$.
\end{itemize}
 {As emphasized in Section \ref{1}, there is no other work on multi-user MIMO RIS-aided systems with FBL coding. Thus, we compare the performance of our proposed algorithm to a single-stream data transmission scheme, the multiple-stream data transmission scheme for systems without RIS, and with a non-optimized RIS coefficients as benchmarks. }
\subsection{Spectral Efficiency Metrics}\label{sec-iv-a}
\begin{figure}
    \centering
    \begin{subfigure}
{0.24\textwidth}
        \centering
           \includegraphics[width=.9\textwidth]{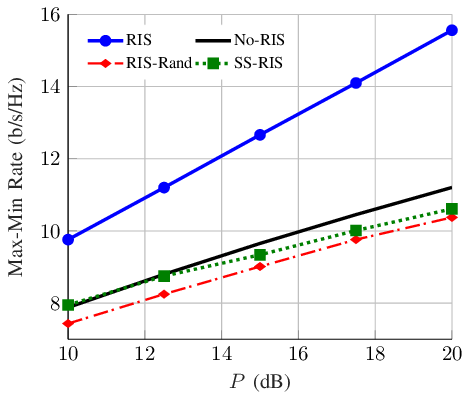}
        \caption{Average max-min rate.}
    \end{subfigure}
\begin{subfigure}
{0.24\textwidth}
        \centering
       \includegraphics[width=.9\textwidth]{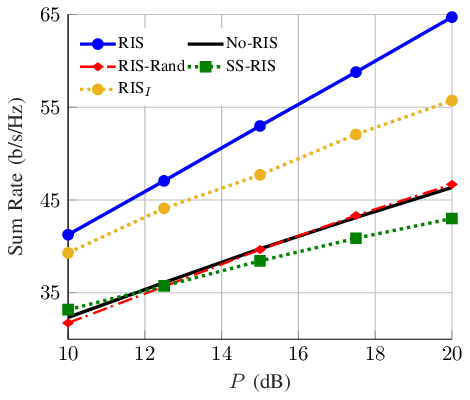}
        \caption{Average sum rate.}
    \end{subfigure}%
    \caption{Spectral efficiency metrics versus $P$ for $N_{{BS}}=4$, $N_{u}=4$, $K=2$, $L=2$,  $M=2$, $n_t=256$, $\epsilon=10^{-5}$,  and $N_{{RIS}}=20$.}
	\label{Fig-rr1}  
\end{figure}

Here, we present numerical results for SE maximization. To this end, we consider the maximum of the average minimum rate and the average sum rate of users as performance metrics. We refer to the maximum of the minimum achievable rate of users as the max-min rate. Note that it is likely that all the users get the same achievable rate when maximizing the minimum rate, which can provide a reasonable fairness among the users \cite{soleymani2020improper}. Nevertheless, when the sum rate is maximized, it could be the case that the users with weaker channels are switched off if QoS constraints are not considered.  
To provide a comprehensive analysis, we explore the impact of various parameters on  the system performance, including the BS power budgets, packet length, as well as the maximum tolerable packet error rate.

\subsubsection{Impact of power budget}

In Fig. \ref{Fig-rr1}, we show the average max-min rate and sum rate versus $P$ for $N_{{BS}}=4$, $N_{u}=4$, $K=2$, $L=2$,  $M=2$, $n_t=256$, $\epsilon=10^{-5}$,  and $N_{{RIS}}=20$. In this figure, the RISs significantly improve the min-max rate and the sum rate, when the RIS elements are optimized. 
 {Surprisingly, employing RISs having random elements degrades the max-min rate. Moreover, we can observe that the benefits of RIS increase with $P$ for these examples.  Additionally, the multi-stream scheme substantially improves the SE for these two examples, where even the No-RIS scheme is slightly better than the single-stream scheme, especially at high SNRs. In these examples, the maximum number of streams per user is $I=\min(N_{{BS}},N_{u})=4$, which makes single-stream transmission gravely suboptimal.  Indeed, when the system complexity increases, a more sophisticated resource allocation scheme should be adopted in order to avoid performance degradation.}

\subsubsection{Impact of packet length}

\begin{figure}
    \centering
    \begin{subfigure}
{0.24\textwidth}
        \centering
\includegraphics[width=\textwidth]{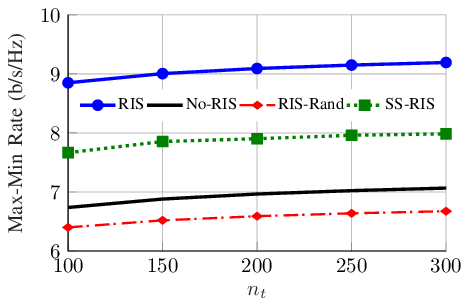}
        \caption{Average max-min rate.}
    \end{subfigure}
\begin{subfigure}
{0.24\textwidth}
        \centering
       \includegraphics[width=.8\textwidth]{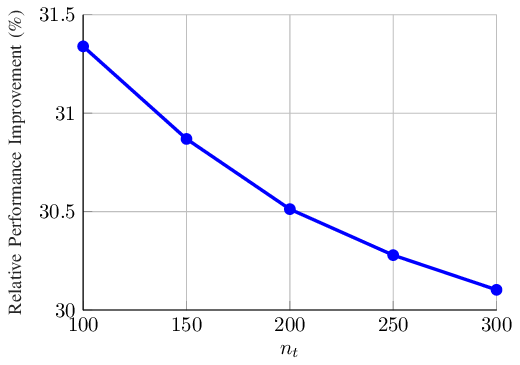}
        \caption{Average performance improvement.}
    \end{subfigure}%
    \caption{Average max-min rate and relative performance improvement by RIS versus $n_t$ for $P=10$ dB, $N_{{BS}}=4$, $N_{u}=3$, $K=2$, $L=2$,  $M=2$, $\epsilon=10^{-5}$,  and $N_{{RIS}}=20$.}  
	\label{Fig-rr2} 
\end{figure}
Fig. \ref{Fig-rr2} shows the average max-min rate and the relative performance improvement by RIS versus  $n_t$ for $P=10$ dB, $N_{{BS}}=4$, $N_{u}=3$, $K=2$, $L=2$,  $M=2$, $\epsilon=10^{-5}$,  and $N_{{RIS}}=20$.
Observe that an RIS can substantially increase the average max-min rate.  {The benefits of deploying RISs decrease with $n_t$. Indeed, the shorter the packet length is, the higher the relative improvements provided by RISs can be. As indicated, the packet length correlates with the level of stringency in the latency constraint. Shorter packet lengths are needed, when the latency constraint  is more stringent. Thus, this result shows that the RIS benefits increase as the latency constraint becomes more stringent. In other words, RISs can even be more beneficial for URLLC systems.}  Moreover, the rates increase with $n_t$ and converge to the Shannon rate when $n_t$ becomes higher. {Furthermore, we can observe that employing multi-stream data transmission substantially increases the average max-min rate for all the values of $n_t$. Note that the benefits of multi-stream data transmission in this figure is lower than in Fig. \ref{Fig-rr1} since the maximum number of streams per users, $I$, is $3$ in this example, which is less than $I$ in Fig. \ref{Fig-rr1}, and the benefits of multi-stream schemes grows with $I$.}

\begin{figure}[t]
    \centering
    \begin{subfigure}[t]{0.4\textwidth}
        \centering
           \includegraphics[width=.9\textwidth]{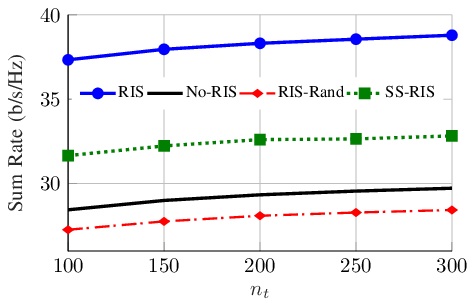}
        \caption{$N_{{BS}}=4$, $N_{u}=3$, $K=2$, and $N_{{RIS}}=20$.}
    \end{subfigure}
\begin{subfigure}[t]{0.4\textwidth}
        \centering
       \includegraphics[width=.9\textwidth]{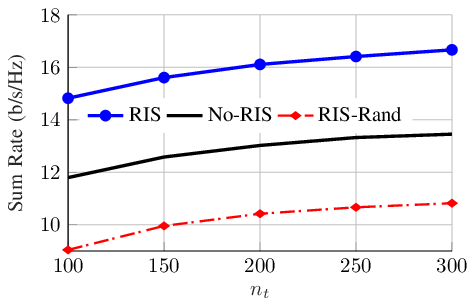}
        \caption{$N_{{BS}}=2$, $N_{u}=2$, $K=4$, and $N_{{RIS}}=60$.}
    \end{subfigure}%
    \caption{Average sum rate versus $n_t$ for $P=10$ dB, $L=2$,  $M=2$, and $\epsilon=10^{-5}$.}  
	\label{Fig-rr3} 
\end{figure}
Fig. \ref{Fig-rr3} shows the average sum rate versus  $n_t$ for $P=10$ dB, $L=2$,  $M=2$, $\epsilon=10^{-5}$ and for different values of $N_{BS}$, $N_u$, $K$ and $N_{RIS}$. In this figure, RISs substantially enhance the average sum rate. However, the benefits of RISs are much more significant when there are less users in the system (Fig. \ref{Fig-rr3}a). Moreover, we can observe that the impact of decreasing the packet length is more severe in the system for a higher number of users. This may show the importance of employing effective interference-management techniques, which should be addressed in future studies. In Fig. \ref{Fig-rr3}a, we also observe that the algorithm conceived for the multi-stream data transmission per user outperforms the beamforming scheme, which employs a single-stream data transmission. 

\subsubsection{Impact of the reliability constraint}

\begin{figure}[t]
    \centering
    \begin{subfigure}[t]{0.24\textwidth}
        \centering
\includegraphics[width=\textwidth]{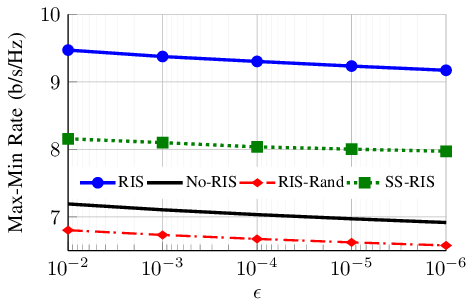}
        \caption{Average max-min rate.}
    \end{subfigure}
\begin{subfigure}[t]{0.24\textwidth}
        \centering
       \includegraphics[width=.8\textwidth]{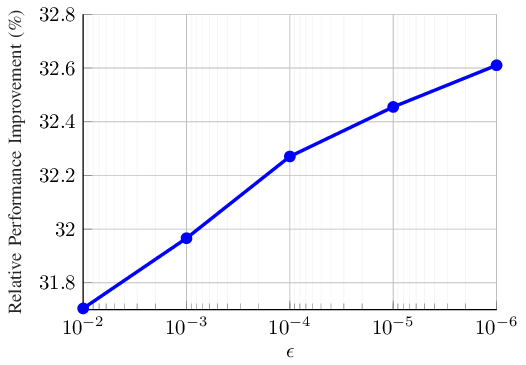}
        \caption{Average performance improvement.}
    \end{subfigure}%
    \caption{Average max-min rate and relative performance improvement by RIS versus $\epsilon$ for $P=10$ dB, $N_{{BS}}=4$, $N_{u}=3$, $K=2$, $L=2$,  $M=2$, $n_t=256$ bits,  and $N_{{RIS}}=20$.}  
	\label{Fig-rr4} 
\end{figure}
Fig. \ref{Fig-rr4} shows the average max-min rate and relative performance improvement enabled by the deployment of RISs versus  $\epsilon$ for $P=10$ dB, $N_{{BS}}=4$, $N_{u}=3$, $K=2$, $L=2$,  $M=2$, $n_t=256$ bits,  and $N_{{RIS}}=20$. In this example, RIS significantly enhances the average max-min rate for all the $\epsilon$ considered. As expected, the average max-min rate decreases when the reliability constraint is more stringent. In other words, we have to transmit at a lower rate to reduce the decoding error rate. Moreover, the multi-stream scheme significantly outperforms the single-stream data transmission. Furthermore, we can observe in Fig. \ref{Fig-rr4}b that the benefits of RIS increase when $\epsilon$ decreases. Thus, RISs can be even more beneficial when more reliable communication is required.

\begin{figure}[t]
    \centering
    \begin{subfigure}[t]{0.4\textwidth}
        \centering
           \includegraphics[width=.9\textwidth]{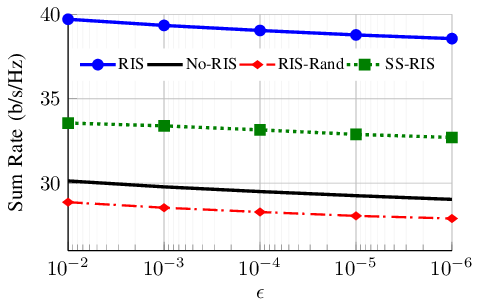}
        \caption{$N_{{BS}}=4$, $N_{u}=3$, $K=2$, and $N_{{RIS}}=20$.}
    \end{subfigure}
\begin{subfigure}[t]{0.4\textwidth}
        \centering       \includegraphics[width=.9\textwidth]{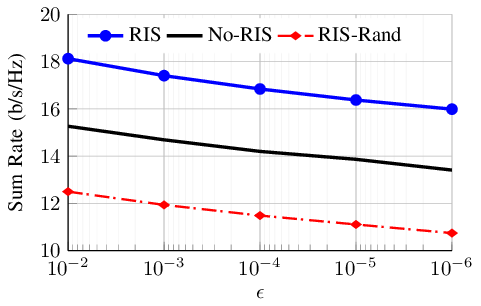}
        \caption{$N_{{BS}}=2$, $N_{u}=2$, $K=4$, and $N_{{RIS}}=60$.}
    \end{subfigure}%
    \caption{Average sum rate versus $\epsilon$ for $P=10$ dB, $L=2$,  $M=2$, and $n_t=256$ bits.}
	\label{Fig-rr5} 
\end{figure}
Fig. \ref{Fig-rr5} shows the average sum rate versus $\epsilon$  for $P=10$ dB, $L=2$,  $M=2$, $n_t=256$ bits and different $N_{BS}$, $N_u$, $K$ and $N_{RIS}$.  As it can be observed, RISs enhance the average sum rate in both networks, if the RIS elements are optimized. However, RISs with random coefficients decrease the average sum rate in these two examples.  {We can also observe that the impact of varying $\epsilon$ is higher in the system supporting more users. Additionally, we can note a significantly greater advantage from RISs when there are fewer users in the network.}
Moreover, we can observe that the average sum rate substantially increases, if we employ multi-stream data transmission.

\subsubsection{Comparison of reflective RIS and STAR-RIS}
\begin{figure}[t]
    \centering
      \includegraphics[width=.4\textwidth]{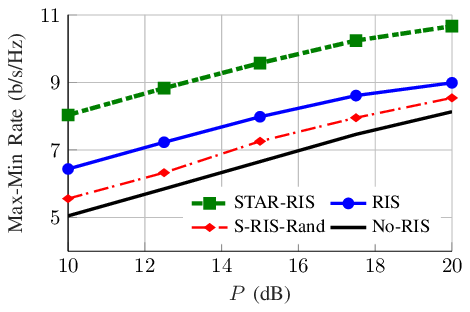}
    \caption{Average max-min rate versus $P$ for $N_{{BS}}=2$, $N_{u}=2$, $K=2$, $L=1$,  $M=1$, $\epsilon=10^{-5}$,  $n_t=256$ bits and $N_{{RIS}}=40$.} 
	\label{Fig-rr6} 
\end{figure}
A reflective RIS has the same performance as a STAR-RIS if all the users are in the reflection half-space of the reflective/STAR-RIS. Thus, to evaluate the performance differences between these two technologies, we consider a single-cell $2\times 2$ MIMO BC in which one of the users is in the reflection space, and the other one is in the transmission space. As shown in Fig. \ref{Fig-rr6}, 
STAR-RIS using the MS scheme can significantly outperform reflective RIS in this example.
For instance, STAR-RIS provides about $25\%$ higher  average max-min rate at $P=10$ dB compared to the reflective RIS.

\subsubsection{ {Convergence behavior}}
\begin{figure}[t]
    \centering
      \includegraphics[width=.4\textwidth]{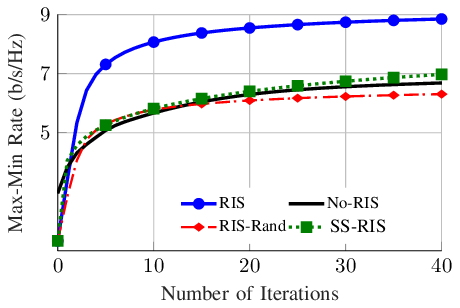}
    \caption{ {Average max-min rate versus the number of iterations for $N_{BS} = 4$, $N_u = 3$,
$K = 2$, $\epsilon=10^{-4}$, $n_t = 256$ bits and $N_{RIS} = 20$.}} 
	\label{Fig-rr7} 
\end{figure}
  {Fig. \ref{Fig-rr7} shows the average max-min
rate versus the number of iterations for $N_{BS} = 4$, $N_u = 3$,
$K = 2$, $\epsilon=10^{-4}$, $n_t = 256$ bits and $N_{RIS} = 20$. This figure illustrates the convergence of the considered algorithms, from which the performance-complexity tradeoff may be inferred. In this example, our scheme proposed for RIS-aided systems employing multiple streams per user outperforms the final solution of the other algorithms after as few as four iterations. Indeed, the lower computational complexity of these alternative algorithms results in significant performance degradation, while to achieve a certain target performance, the RIS scheme requires substantially fewer iterations. Additionally, the  No-RIS, RIS-Rand, and SS-RIS algorithms may not support a specific max-min rate, while our framework can achieve it with a reasonable number of iterations. For instance, while the No-RIS, RIS-Rand, and SS-RIS  algorithms struggle to support a max-min rate of 8 b/s/Hz, our framework achieves it in only nine iterations.}

\subsection{Energy Efficiency Metrics}\label{sec-iv-b}
Now we investigate the EE of RIS in MU-MIMO URLLC BCs. To this end, we consider the impact of $P_c$, $\epsilon$, and $n_t$. In the examples provided in this subsection, we assume that each RIS consumes $1$ W power. Thus, to make a fair comparison, we consider a lower constant power ($P_c$) for the systems operating without RISs. Moreover, we assume that the power budget of the BSs is $P=10$ dB. 
\begin{figure}[t]
    \centering
    \begin{subfigure}[t]{0.24\textwidth}
        \centering
           \includegraphics[width=\textwidth]{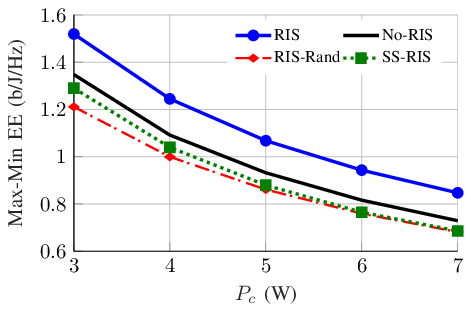}
        \caption{Average max-min EE.}
    \end{subfigure}
\begin{subfigure}[t]{0.24\textwidth}
        \centering
       \includegraphics[width=\textwidth]{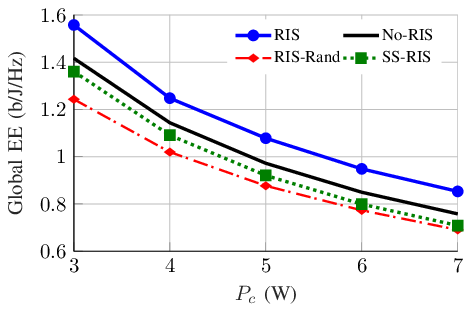}
        \caption{Average global EE.}
    \end{subfigure}%
    \caption{Energy efficiency metrics versus $P_c$ for $N_{{BS}}=5$, $N_{u}=5$, $K=2$, $L=2$,  $M=2$, $n_t=256$, $\epsilon=10^{-5}$,  and $N_{{RIS}}=20$.}
	\label{Fig-ee}  
\end{figure}
\subsubsection{Impact of $P_c$} Fig. \ref{Fig-ee} shows the average max-min EE and GEE versus $P_c$ for $N_{{BS}}=5$, $N_{u}=5$, $K=2$, $L=2$,  $M=2$, $n_t=256$ bits, $\epsilon=10^{-5}$, and $N_{{RIS}}=20$.  Note that we use the term max-min EE to refer to the highest minimum EE. As it can be observed, RISs significantly increase the average max-min EE and GEE with FBL. Interestingly, RISs may reduce the EE, if their elements are random. However, our proposed algorithms can statistically enhance the EE of RIS-aided scenarios. 
For instance, in the particular example of Fig. \ref{Fig-ee}a, an RIS provides more than $10\%$ improvements over the systems disregarding the RISs for all the  values of $P_c$ considered.
Additionally, the multi-stream scheme significantly outperforms the single-stream data transmission in the both examples. 

\subsubsection{Impact of the reliability constraint}
\begin{figure}[t]
        \centering
           \includegraphics[width=.4\textwidth]{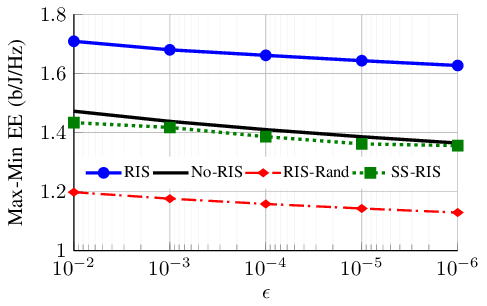} 
    \caption{Average max-min EE  versus $\epsilon$ for $P=10$ dB, $N_{{BS}}=4$, $N_{u}=4$, $K=2$, $L=1$,  $M=1$, $n_t=256$ bits,  and $N_{{RIS}}=20$.}
	\label{Fig-ee2}  
\end{figure}
Fig. \ref{Fig-ee2} shows the average max-min EE  versus $\epsilon$ for $P=10$ dB, $N_{{BS}}=4$, $N_{u}=4$, $K=2$, $L=1$,  $M=1$, $n_t=256$ bits,  and $N_{{RIS}}=20$. In this example, the RIS provides a significant gain, which increases with $\epsilon$. Again, RIS decreases the max-min EE when its elements are random. 
Moreover, we can observe that single-stream data transmission is suboptimal in this $4\times 4$ MIMO system. Indeed,   the multi-stream systems communicating without RIS outperforms the RIS-aided single-stream data transmission. 

\begin{figure}[t]
    \centering
       \includegraphics[width=.3\textwidth]{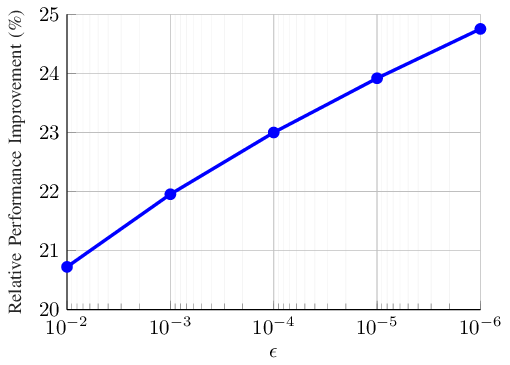} 
		    \caption{Average EE performance improvement by RIS versus $\epsilon$ for $P=10$ dB, $N_{{BS}}=2$, $N_{u}=2$, $K=2$, $L=2$,  $M=2$, $n_t=256$ bits,  and $N_{{RIS}}=20$.}
	\label{Fig-ee3}  
\end{figure}
Fig. \ref{Fig-ee3} shows the average EE performance improvement attained by RISs versus $\epsilon$ for $P=10$ dB, $N_{{BS}}=2$, $N_{u}=2$, $K=2$, $L=2$,  $M=2$, $n_t=256$ bits,  and $N_{{RIS}}=20$.
 {In this example, the RIS substantially increases the average max-min for all values of $\epsilon$. 
Moreover,  higher gains are achieved by RISs, when the tolerable bit error rate is lower. Thus, the more reliable the communication has to be, the more energy efficient the RIS-aided systems becomes.} 

\subsubsection{Impact of the packet length}
\begin{figure}[t]
    \centering
        \includegraphics[width=.4\textwidth]{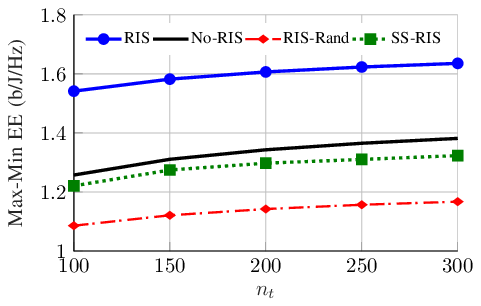}
    \caption{Average max-min EE versus $n_t$ for $P=10$ dB, $N_{{BS}}=4$, $N_{u}=4$, $K=2$, $L=1$,  $M=1$, $\epsilon=10^{-5}$,  and $N_{{RIS}}=20$.}
	\label{Fig-ee4}  
\end{figure}
Fig. \ref{Fig-ee4} shows the average max-min EE versus $n_t$ for $P=10$ dB, $N_{{BS}}=4$, $N_{u}=4$, $K=2$, $L=1$,  $M=1$, $\epsilon=10^{-5}$,  and $N_{{RIS}}=20$.
In this figure, RIS provides a significant gain when the RIS elements are optimized by our proposed algorithm. However, random RIS coefficients degrade  the EE performance. 

\begin{figure}[t]
    \centering
       \includegraphics[width=.3\textwidth]{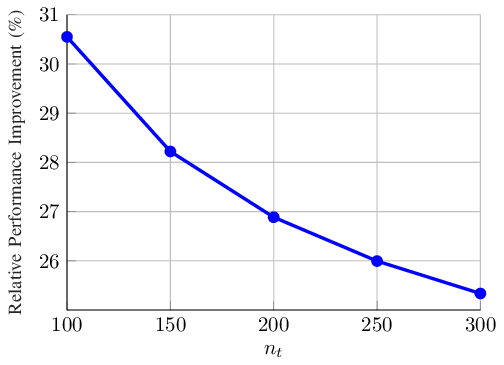}
		    \caption{Average EE performance improvement by RIS versus $n_t$ for $P=10$ dB, $N_{{BS}}=2$, $N_{u}=2$, $K=2$, $L=2$,  $M=2$, $\epsilon=10^{-5}$,  and $N_{{RIS}}=20$.}
	\label{Fig-ee5}  
\end{figure}
Fig. \ref{Fig-ee5} shows the average EE improvement by RISs versus $n_t$ for $P=10$ dB, $N_{{BS}}=2$, $N_{u}=2$, $K=2$, $L=2$,  $M=2$, $\epsilon=10^{-5}$,  and $N_{{RIS}}=20$.
 {The average improvements reduce as $n_t$ increases. This indicates that the lower the tolerable latency, the higher gain the RIS can provide. In other words, RIS-aided systems become more energy efficient when a low latency is required, as in control channels, for example.}

\section{ {Summary, Conclusions, and Future Research}}\label{sec-v}
An optimization framework was  proposed for MU-MIMO RIS-aided systems with FBL by considering the NA for the rate expressions. 
To this end, we first calculated closed-form expressions for the FBL rate and then obtained suitable concave lower bounds for the FBL rates. 
Our proposed framework can be adapted to a large variety of MU-MIMO systems in which interference is treated as noise. Moreover, the framework can obtain a stationary point of a broad spectrum of practical optimization problems such as the maximization of the minimum/sum rate, GEE and minimum EE, when the set of the feasible RIS coefficients adheres to convexity. In summary, the key conclusions of this work are:
\begin{itemize}
\item RISs may significantly increase the average max-min rate, sum rate, max-min EE and global EE of the MU-MIMO systems considered. However, the RIS elements should be optimized to attain the above benefits, since RISs utilizing random elements may even degrade the system performance.

\item  The benefits of RISs increase when the packet length is reduced and/or the tolerable bit error rate is lower. The packet length can be related to the latency constraint, and the tolerable bit error rate represents the reliability constraint. Thus, these results show that RISs can be even more beneficial in URLLC systems than in non-URLLC systems.

\item  Multiple-stream data transmission for each user significantly outperforms single-stream data transmission (beamforming) in MU-MIMO RIS-assisted URLLC systems. Indeed, both the reliability and latency can be enhanced, when multiple-stream data transmission is employed in multiple-antenna systems.
\end{itemize}

  {In future research, it would be interesting to integrate advanced interference-management techniques, such as RSMA and NOMA, into MU-MIMO URLLC systems. In this regard, the solutions in, e.g., \cite{soleymani2022rate, soleymani2023noma}, in combinations with the FBL rate expressions in Lemma \ref{lem-r} might be helpful. Another challenging research direction is to extend the optimization framework proposed in this work to scenarios with statistical or imperfect CSI. To this end, the robust designs in \cite{yao2023robust, zhou2020framework} can be harnessed.  Moreover, studying the performance of other concepts/technologies for RIS, such as holographic RIS \cite{an2023stacked}, active RIS \cite{zhang2022active}, BD-RIS \cite{li2022beyond, santamaria2024mimo}, and globally-passive RIS \cite{fotock2023energy, soleymani2024energy} can be another promising direction for extending this work.  Furthermore, considering RIS-aided cell-free URLLC systems is  worth exploring in future research. Finally, another interesting line of research is to develop  computationally more efficient resource allocation schemes that do not compromise performance. 
} 
\bibliographystyle{IEEEtran}
\bibliography{ref2}

\begin{IEEEbiography}[{\includegraphics[width=1in,height=1.25in,clip,keepaspectratio]{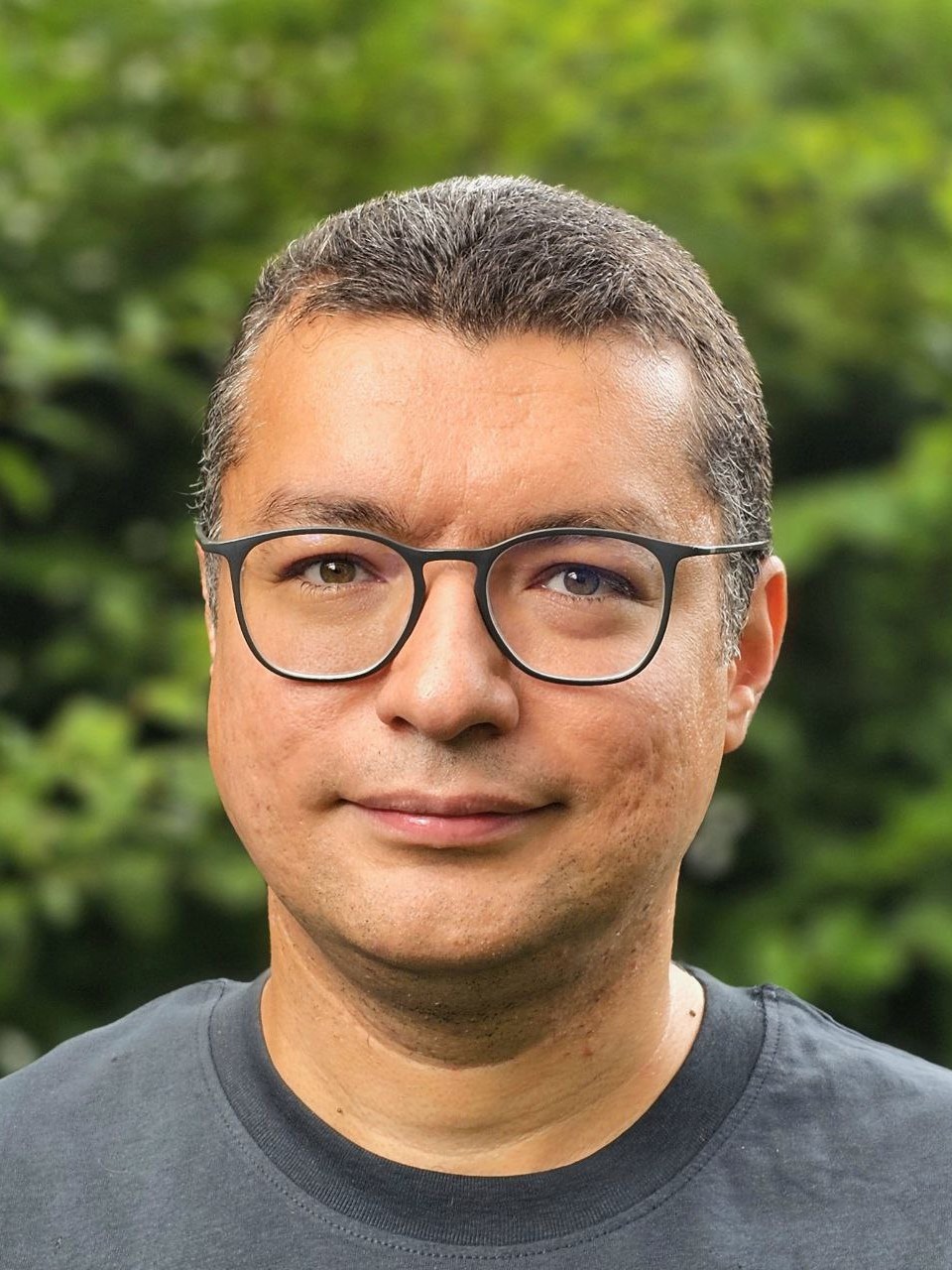}}]
{Mohammad Soleymani} was born in Arak, Iran. He received the B.Sc. degree from Amirkabir University of Technology (Tehran Polytechnic), the M.Sc. degree from Sharif University of Technology, Tehran, Iran, and the Ph.D. degree (with distinction) from the University of Paderborn, Germany, all in electrical engineering. He is currently an assistant professor (Akademischer Rat) at the Signal and System Theory Group, University of Paderborn. He was a Visiting Researcher at the University of Cantabria, Spain. His research interests include multiuser communications, wireless networking, convex optimization and statistical signal processing.
\end{IEEEbiography} 

\begin{IEEEbiography}[{\includegraphics[width=1in,height=1.25in,clip,keepaspectratio]{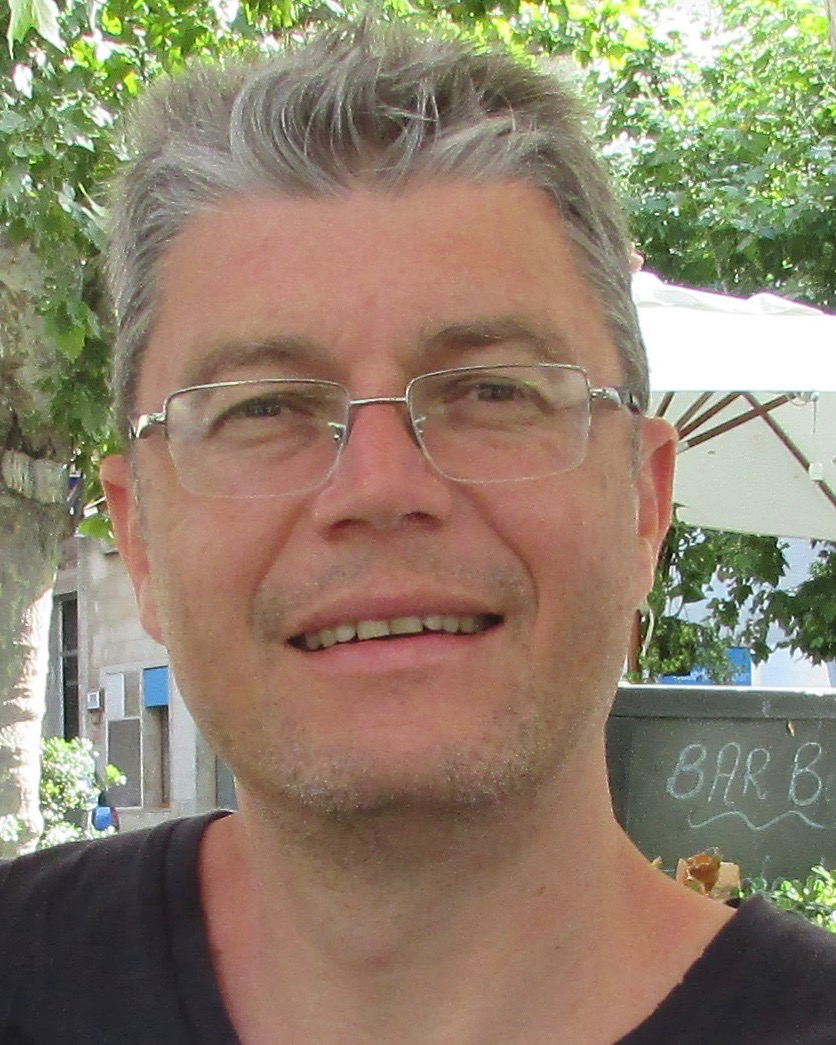}}]{Ignacio Santamaria} (M'96-SM'05) 
received the Telecommunication Engineer degree and the Ph.D. degree in electrical engineering from the Universidad Politecnica de Madrid (UPM), Spain, in 1991 and 1995, respectively. In 1992, he joined the Department of Communications Engineering, Universidad de Cantabria, Spain, where he is Full Professor since 2007. He has co-authored more than 250 publications in refereed journals and international conference papers, and holds two patents. He has co-authored the book D. Ramirez, I. Santamaria, and L.L. Scharf, ``Coherence in Signal Processing and Machine Learning'', Springer, 2022. His current research interests include signal processing algorithms and information-theoretic aspects of multiuser multiantenna wireless communication systems, multivariate statistical techniques and machine learning theories. He has been involved in numerous national and international research projects on these topics.  He has been a visiting researcher at the University of Florida (in 2000 and 2004), at the University of Texas at Austin (in 2009), and at the Colorado State University (in 2015 and 2017). He has been Associate Editor of the IEEE Transactions on Signal Processing (2011-2015), and Senior Area Editor of the IEEE Transactions on Signal Processing (2013-2015). He has been a member of the IEEE Machine Learning for Signal Processing Technical Committee (2009-2014), member of the IEEE Signal Processing Theory and Methods Technical Committee (2020-2022), and member of the IEEE Data Science Initiative (DSI) steering committee (2020-2022). Prof. Santamaria co-authored a paper that received the 2012 IEEE Signal Processing Society Young Author Best Paper Award, and has received the 2022 IHP International Wolfgang Mehr Fellowship Award.
\end{IEEEbiography} 

\begin{IEEEbiography}[{\includegraphics[width=1in,height=1.25in,clip,keepaspectratio]{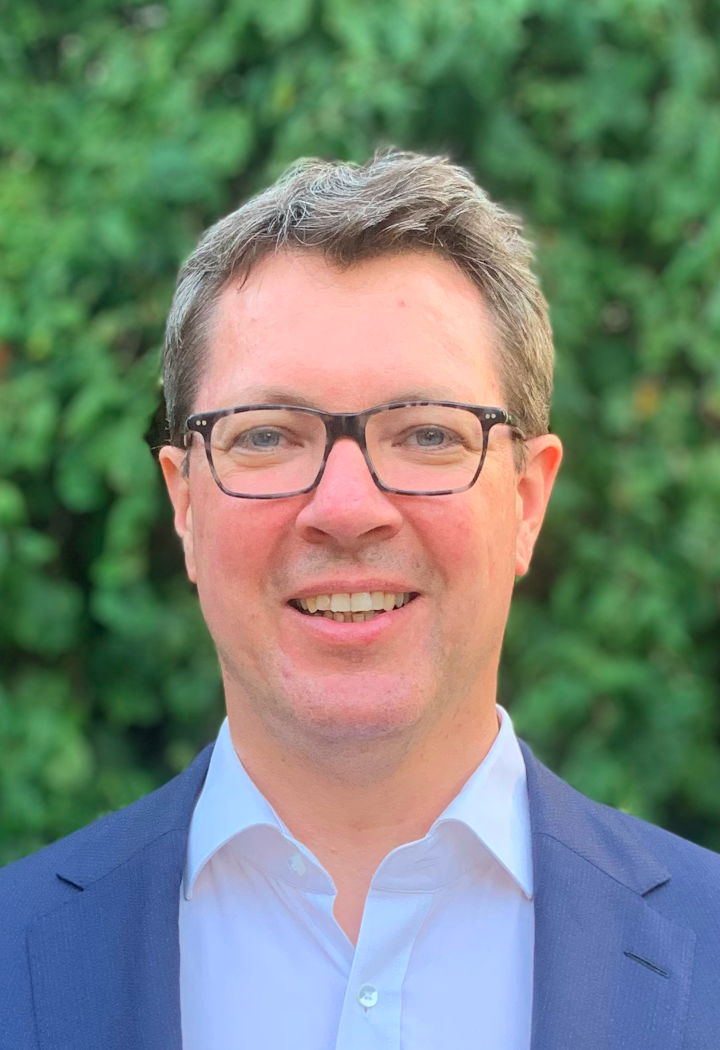}}]
{Eduard A. Jorswieck} (Fellow, IEEE) received the Ph.D. degree in computer engineering from TU Berlin, in 2004. From 2006 to 2008, he was a Postdoctoral Fellow and an Assistant Professor with the Signal Processing Group, KTH Stockholm. From 2008 to 2019, he was the Chair for Communication Theory with TU Dresden. He is currently the Managing Director of the Institute of Communications Technology and the Head of the Chair for Communications Systems and a Full Professor with Technische  Universit\"at  Braunschweig, Germany. His research interest lies in the broad area of communications. He has published more than 170 journal articles, 15 book chapters, one book, three monographs, and more than 320 conference papers. He is a Fellow of the IEEE. He was a recipient of the IEEE Signal Processing Society Best Paper Award. Since 2017, he has been the Editor-in-Chief of the Springer EURASIP Journal on Wireless Communications and Networking. Since 2022, he has been on the editorial board of the IEEE TRANSACTIONS ON COMMUNICATIONS. He was on the editorial boards of the IEEE SIGNAL PROCESSING LETTERS, the IEEE TRANSACTIONS ON SIGNAL PROCESSING, the IEEE TRANSACTIONS ON WIRELESS COMMUNICATIONS, and the IEEE TRANSACTIONS ON INFORMATION FORENSICS AND SECURITY.
\end{IEEEbiography}

\begin{IEEEbiography}[{\includegraphics[width=1in,height=1.25in,clip,keepaspectratio]{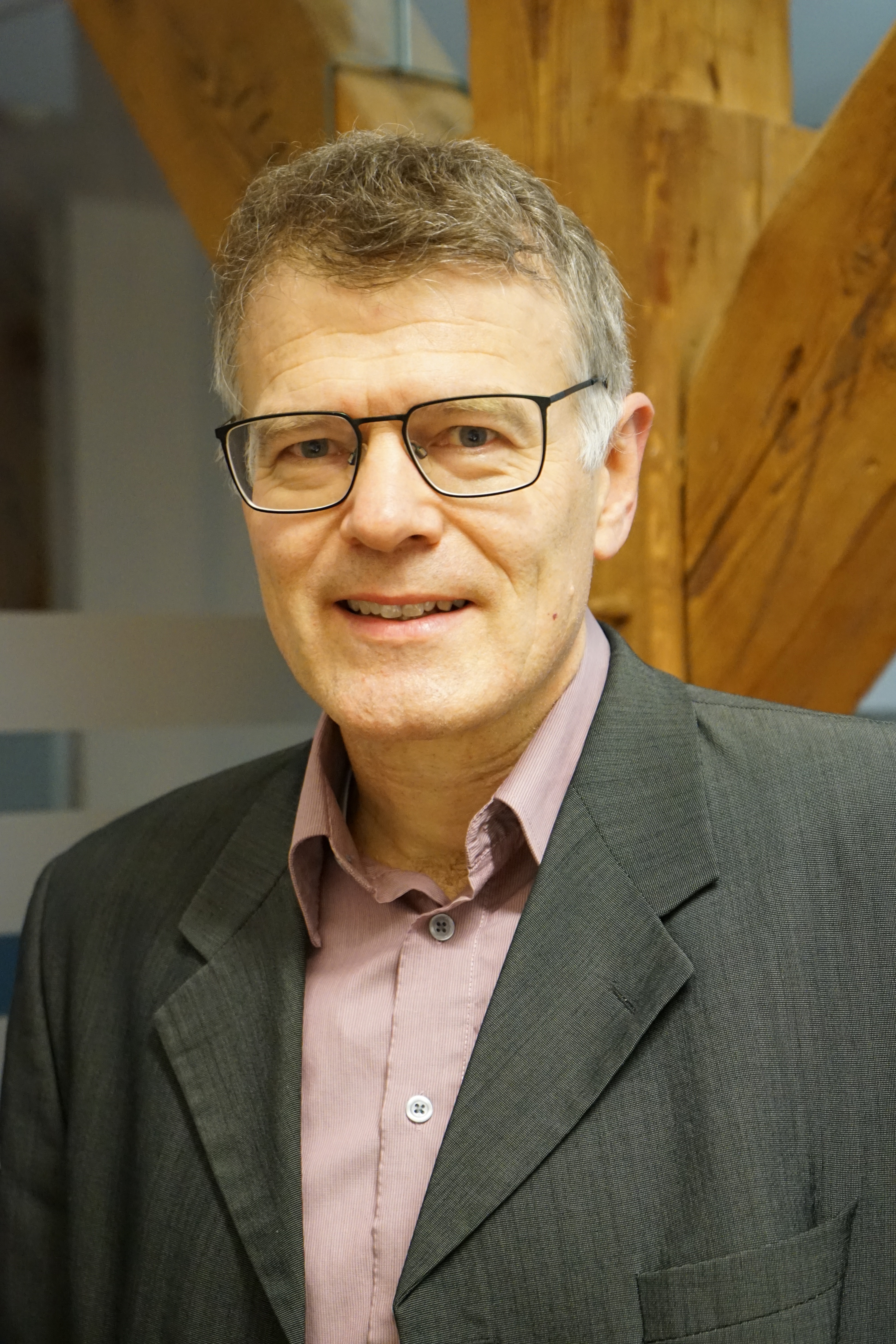}}]
{Robert Schober} (S'98, M'01, SM'08, F'10) received the Diplom (Univ.) and the Ph.D. degrees in electrical engineering from Friedrich-Alexander University of Erlangen-Nuremberg (FAU), Germany, in 1997 and 2000, respectively. From 2002 to 2011, he was a Professor and Canada Research Chair at the University of British Columbia (UBC), Vancouver, Canada. Since January 2012 he is an Alexander von Humboldt Professor and the Chair for Digital Communication at FAU. His research interests fall into the broad areas of Communication Theory, Wireless and Molecular Communications, and Statistical Signal Processing.

Robert received several awards for his work including the 2002 Heinz Maier­ Leibnitz Award of the German Science Foundation (DFG), the 2004 Innovations Award of the Vodafone Foundation for Research in Mobile Communications, a 2006 UBC Killam Research Prize, a 2007 Wilhelm Friedrich Bessel Research Award of the Alexander von Humboldt Foundation, the 2008 Charles McDowell Award for Excellence in Research from UBC, a 2011 Alexander von Humboldt Professorship, a 2012 NSERC E.W.R. Stacie Fellowship, a 2017 Wireless Communications Recognition Award by the IEEE Wireless Communications Technical Committee, and the 2022 IEEE Vehicular Technology Society Stuart F. Meyer Memorial Award. Furthermore, he received numerous Best Paper Awards for his work including the 2022 ComSoc Stephen O. Rice Prize and the 2023 ComSoc Leonard G. Abraham Prize. Since 2017, he has been listed as a Highly Cited Researcher by the Web of Science. Robert is a Fellow of the Canadian Academy of Engineering, a Fellow of the Engineering Institute of Canada, and a Member of the German National Academy of Science and Engineering.

He served as Editor-in-Chief of the IEEE Transactions on Communications, VP Publications of the IEEE Communication Society (ComSoc), ComSoc Member at Large, and ComSoc Treasurer. Currently, he serves as Senior Editor of the Proceedings of the IEEE and as ComSoc President.

\end{IEEEbiography} 

\begin{IEEEbiography}[{\includegraphics[width=1in,height=1.25in,clip,keepaspectratio]{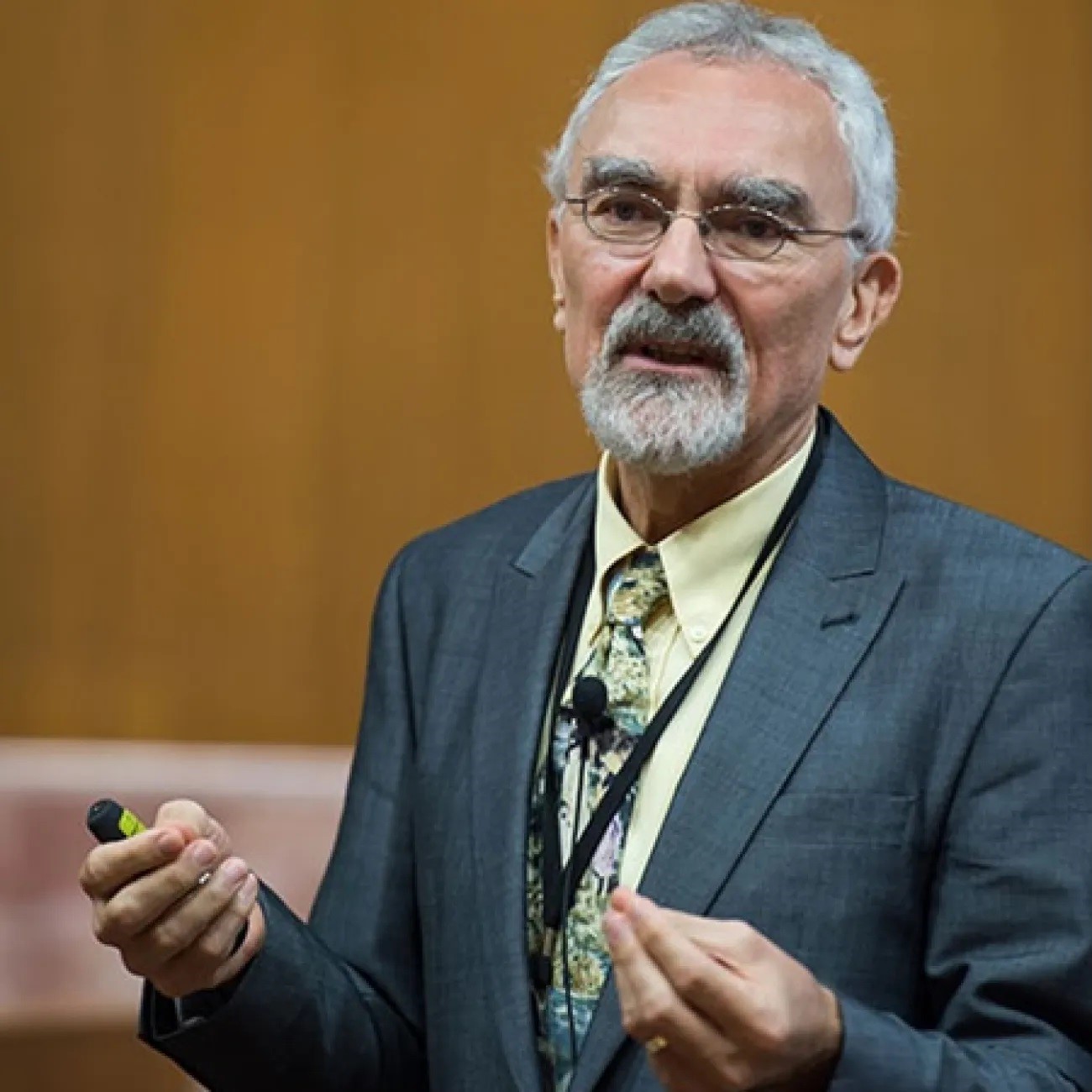}}]
{Lajos Hanzo}  is a Fellow of the Royal Academy of Engineering, FIEEE, FIET, Fellow of EURASIP and a Foreign Member of the Hungarian Academy of Sciences. He coauthored 2000+ contributions at IEEE Xplore and 19 Wiley-IEEE Press monographs. He was bestowed upon the IEEE Eric Sumner Technical Field Award.

\end{IEEEbiography} 

\end{document}